\documentclass[11pt]{article}

\usepackage{amsfonts,amssymb,amsmath,latexsym,ae,aecompl}

\newcommand{\FFF}{\vspace*{\bigskipamount}}

\newcommand{\BBB}{\vspace*{-\bigskipamount}}

\newcommand{\cA}{\mathcal{A}}

\newcommand{\cE}{\mathcal{E}}

\newcommand{\cI}{\mathcal{I}}

\newcommand{\cO}{\mathcal{O}}
\newcommand{\cP}{\mathcal{P}}

\newcommand{\cS}{\mathcal{S}}

\newcommand{\cU}{\mathcal{U}}

\newcommand{\cW}{\mathcal{W}}

\newcommand{\Paragraph}[1]{\BBB\paragraph{#1}}
\newcommand{\remove}[1]{}

\newlength{\pagewidth}
\newlength{\captionwidth}
\newcommand{\qed}{\hfill $\square$ \smallbreak}
\newenvironment{proof}{\noindent{\bf Proof:}}{\qed}

\newtheorem{theorem}{Theorem}
\newtheorem{lemma}{Lemma}

\newtheorem{corollary}{Corollary}
\newtheorem{proposition}{Proposition}

\begin{document}

\baselineskip   	3ex
\parskip        	1ex

\thispagestyle{empty}

\title{		Maximum Throughput of Multiple Access Channels\\
		in Adversarial Environments \footnotemark[1]\FFF\FFF}

\author{	Bogdan S. Chlebus  	\footnotemark[2]   	\and
		Dariusz R. Kowalski 	\footnotemark[3]	\and
		Mariusz A. Rokicki 	\footnotemark[4]}

\footnotetext[1]{
This paper was published as~\cite{ChlebusKR-DC09}.
The results presented here appeared in a preliminary form in~\cite{ChlebusKR-SSS07} and also partly in~\cite{ChlebusKR-PODC06}.
}

\footnotetext[2]{
               	Department of Computer Science and Engineering,
               	University of Colorado Denver,
               	Denver, CO 80217, U.S.A.
		Work supported by the National Science Foundation under Grant No. 0310503.}

\footnotetext[3]{
               	Department of Computer Science,
               	University of Liverpool,
               	Liverpool L69 3BX, U.K.
		Work supported by the Engineering and Physical Sciences Research Council Grant EP/G023018/1.}

\footnotetext[4]{
		Work done while the author was a doctoral student at the University of Colorado Denver, supported by the National Science Foundation under Grant No. 0310503, and next a post-doctoral fellow at the Centre National de la Recherche Scientifique, Universit\' e Paris Sud.}

\date{}

\maketitle

\vfill


\begin{abstract}
We consider deterministic distributed broadcasting on multiple access channels in the framework of adversarial queuing.
Packets are injected dynamically by an adversary that is constrained by the injection rate and the number of packets that may be injected simultaneously; the latter we call burstiness.
An algorithm is stable when the number of packets in queues at the stations stays bounded.
The maximum injection rate that an algorithm can handle in a stable manner is called the throughput of the algorithm.
We consider adversaries of injection rate~$1$, that is, of one packet per round, to address the question if the maximum throughput~$1$ can be achieved, and if so then with what quality of service.
We develop an algorithm that achieves throughput~$1$ for any number of stations against leaky-bucket adversaries.
The algorithm has $\cO(n^2+\text{burstiness})$ packets queued simultaneously at any time, where $n$ is the number of stations; this upper bound is proved to be best possible.
An algorithm is called fair when each packet is eventually broadcast.
We show that no algorithm can be both stable \emph{and} fair for a system of at least two stations against leaky-bucket adversaries.
We study in detail small systems of exactly two and three stations against window adversaries to exhibit differences in quality of broadcast among classes of algorithms.
An algorithm is said to have fair latency if the waiting time of packets is $\cO(\text{burstiness})$.
For two stations, we show that fair latency can be achieved by a full sensing algorithm, while there is no stable acknowledgment based algorithm.
For three stations, we show that fair latency can be achieved by a general algorithm, while no full sensing algorithm can be stable.
Finally, we show that algorithms that either are fair or do not have the queue sizes affect the order of transmissions cannot be stable in systems of at least four stations against window adversaries.

\vfill

\end{abstract}

\vfill

\thispagestyle{empty}

\setcounter{page}{0}


\newpage

\section{Introduction}

\label{sec:introduction}

Multiple access channels model distributed communication environments supporting instantaneous broadcasting.
The properties of a system consisting of a number of stations attached to a transmission medium that make it a multiple access channel are twofold:
1) a packet transmitted by a station reaches all the stations instantaneously; and
2) a packet is successfully received if its transmission does not overlap with any other transmissions.

We restrict attention to synchronous ``slotted'' model in which stations use local clocks ticking at the same rate and indicating the same round numbers.
A station transmits at a round determined by its clock, with a transmission filling the whole round.
This means that if at least two stations transmit at a round, then no messages are received at this round, otherwise the sole transmission reaches every station instantaneously.

We consider deterministic distributed broadcasting on multiple access channels in the framework of adversarial queuing.
Packets are injected dynamically by an adversary that is constrained by two parameters: the injection rate  and the number of packets that can be injected simultaneously, which we call burstiness.
Adversaries have injection rate~$1$, that is, of up to one packet per round.
An algorithm is stable when the number of packets stays bounded at all times.
The maximum injection rate that an algorithm can handle in a stable manner is called the throughput of the algorithm.
We address the question if maximum throughput~$1$ can be achieved, and if so then what can be the quality of service.

Fairness of an algorithm denotes the property that each packet is eventually successfully broadcast.
We say that an algorithm is of fair latency, for injection rate~$1$, when the packet latency is $\cO(\text{burstiness})$.
We study the issue of the quality of service with throughput~$1$ depending on the following parameters: the class of adversaries, the subclass of algorithms, and the number~$n$ of stations attached to the channel.
We consider two standard adversarial models: window adversaries and leaky-bucket adversaries, all adversaries with injection rate~$1$ only.
We study two subclasses of algorithms: acknowledgment based and full sensing.

For leaky-bucket adversaries, achieving throughput~$1$ \emph{and} fairness is impossible, except for the trivial case of a single station.
Stability alone, with respect to leaky-bucket adversaries, is not achievable by full sensing algorithms but is achievable by general algorithms.
A stable algorithm of throughput~$1$ we develop has stations store $\cO(n^2+\text{burstiness})$ packets in queues; this upper bound is shown to be asymptotically best possible.

For window adversaries, the situation is more complex.
If there are only two stations, then fair latency is achievable by full sensing algorithms, but no acknowledgment based algorithm can be stable.
For three stations, fair latency is achievable, but no algorithm that is either full sensing nor that withholds the channel for a sequence of exclusive broadcasts after a successful transmission can be stable against leaky-bucket adversaries.
We show that algorithms that are stable against window adversaries in systems of at least four stations should be sufficiently liberal in their design.
In particular, it is impossible for an algorithm to be stable when it simultaneously is either fair or does not have stations use the queue sizes in an effective way or finally has stations behave greedily by withholding the channel after a successful transmission.


\Paragraph{Related work.}

Most of the previous work on dynamic broadcasting on multiple access channels has concentrated on scenarios when packets are injected subject to statistical constraints.
When a broadcast environment is of such nature, the behavior of the system can be modeled as a Markov chain and stability can be captured by ergodicity,  see~\cite{MitzenmacherUpfal-book17} for a probabilistic background.
The well known algorithms like Aloha~\cite{Abramson-TIT85} and binary exponential backoff~\cite{MetcalfeB76} have been investigated with respect to their ability to handle broadcast with stochastic injection rates; Gallager~\cite{Gallager85} gives an overview of the early research in this direction.
For recent work, see the papers by Goldberg et al.~\cite{GoldbergJKP04,GoldbergMPS-JACM00}, H\aa stad et al.~\cite{HastadLR-SICOMP96}, and Raghavan and Upfal~\cite{RaghavanU-SICOMP98}.

Adversarial queuing  was proposed by Borodin et al.~\cite{BorodinKRSW-JACM01} as an approach to study stability of contention-resolution algorithms in store-and-forward routing.
They showed, among other things, that a directed acyclic network is stable with injection rate~1, for any greedy contention-resolution algorithm.
Universal stability of an algorithm denotes stability in any network, and universal stability of a network denotes stability of an arbitrary algorithm in the network, both under injections at a constant rate less than~$1$.
These notions were introduced by Andrews et al.~\cite{AndrewsAFLLK-JACM01}; they were later studied by Gamarnick~\cite{Gamarnik-SICOMP03} and Alvarez et al.~\cite{AlvarezBS-SICOMP04}.
Bhattacharjee et al.~\cite{BhattacharjeeGL-SICOMP04} showed that the natural FIFO algorithm can be unstable at arbitrarily low injection rates.
Lotker et al.~\cite{LotkerPR04} showed that every work-preserving contention-resolution algorithm is stable if injection rate is less than $1/(D+1)$, where $D$ is an upper bound on the length of any path that a packet needs to traverse.
Koukopoulos et al.~\cite{KoukopoulosMNS-TCSy05}  addressed the question of how structural properties of networks affect stability of contention-resolution algorithms.
Adaptive algorithms have packets carry only their destination addresses, rather than complete routing paths; the stability of such algorithms was considered by Aiello et al.~\cite{AielloKOR-JCSS00}.

Bender et al.~\cite{BenderFHKL-SPAA05} studied stability of randomized backoff on multiple access channels in an adversarial setting, where stability meant that throughput was as large as injection rate.
They showed, among other things, that the exponential backoff is unstable for rates $\rho\ge c \lg\lg n/\lg n$, for a sufficiently large constant~$c$.
Stability of deterministic broadcast algorithms for multiple access channels in the framework of adversarial queueing was first considered by Chlebus et al.~\cite{ChlebusKR-TALG12}.
They defined fair latency to hold when the packet latency was $\cO(\textrm{burstiness}/\text{rate})$.
Fair latency implies strong stability, which holds when the number of queued packets is of the order of burstiness.
They showed that no algorithm can be strongly stable for injection rates that are $\omega(\frac{1}{\log n})$ and gave a full sensing algorithm for a channel with collision detection that is  both universally stable and of fair latency for injection rates at most $\frac{1}{2(\lceil\lg n\rceil+1)}$.
For a channel without collision detection, they developed a full sensing algorithm that is both universally stable and of fair latency for injection rates at most $\frac{1}{c \lg^2 n}$,  for some $c>0$.
They showed the existence of an acknowledgment based algorithm that has fair latency for injection rates at most $\frac{1}{c n\lg^2 n}$, for some $c>0$, and developed an explicit acknowledgment based algorithm that has fair latency for injection rates at most $\frac{1}{27 n^2\ln n}$.
Finally, they showed that no acknowledgement-based algorithm is stable for injection rates larger than $\frac{3}{1+\lg n}$.

Algorithmic problems not regarding pure communication issues in distributed environments relying on multiple access channels have been also studied in the literature.
Such work included broadcasting spanning forests when the edges of an input graph are stored in the stations~\cite{ChlebusGK-TCS03} and performing a set of independent unit-cost tasks~\cite{ChlebusKL-DC06}.


\newcommand{\RB}{\raisebox{3ex}{~}}
\newcommand{\LB}{\raisebox{-1.5ex}{~}}

\begin{table}
\begin{center}
\begin{tabular}{|c ||c |c |c |}
\hline
\RB \LB
& $\mathsf{n=2}$ & $\mathsf{n=3}$ & $\mathsf{n \ge 4}$ \\
\hline\hline
\RB \LB
\textsf{ack based} & stable : impossible && \\
\hline
\RB \LB
\textsf{full sensing} &fair latency : possible & stable : impossible&\\
\hline
\RB
\textsf{\raisebox{-1.9ex}{general}} && \textrm{\raisebox{-1.9ex}{fair latency : possible}}
& stable : possible \\
&&& \textrm{\raisebox{0.9ex}{stable and fair : impossible}}\\
\hline
\end{tabular}
\parbox{\pagewidth}{
~
\caption{\label{tab:window-adversaries}
Window adversaries: some possibility and impossibility facts regarding the quality of service of algorithms, depending on a number~$n$ of stations and a subclass of algorithms.
}}
\end{center}
\end{table}


\Paragraph{Summary of the results.}

We have arranged  the main possibility and  impossibility facts in two Tables~\ref{tab:window-adversaries} and~\ref{tab:leaky-bucket-adversaries}, which correspond to the two adversarial models.
The categorization of these facts is with respect to two subclasses of general algorithms: acknowledgment based and full sensing ones, and the number~$n$ of stations in a system.
The entries that are not empty represent theorems in this paper.

Questions regarding existence of algorithms of suitable quality in environments represented by  empty entries in Tables~\ref{tab:window-adversaries} and~\ref{tab:leaky-bucket-adversaries} can be settled by inferring answers from the non-empty  entries of the tables.
The following implications about dependencies among the entries of the two tables hold true a fortiori:
Impossibility for window  adversaries implies the same result for leaky-bucket ones.
An existence of an algorithm of some quality of service for leaky-bucket adversaries implies the  same existence for window ones.
Impossibility for $n$~stations implies the same result for more than $n$~stations.
An existence of an algorithm of some quality of service for $n$~stations implies the same existence for less than $n$~stations.
An existence of an acknowledgement-based algorithm with some properties implies the existence of a full sensing one with the same properties, and existence of a full sensing  algorithm with some properties implies existence of a general one with the same properties.
An impossibility result for general algorithms implies the corresponding result for full sensing ones, and impossibility result for full sensing algorithms implies the corresponding result for acknowledgement-based ones.

Our main positive contribution is discovery of a stable algorithm that maintains $\cO(n^2+\text{burstiness})$ packets in queues at any round against leaky-bucket adversaries, for any number~$n$ of stations.
We next show that any broadcast algorithm for a system of $n$ stations can be forced by a leaky-bucket adversary of burstiness~$2$ to eventually  have $\Omega(n^2)$ packets in queues, thus proving the optimality of our algorithm with respect to the asymptotic number of queued packets.
Except for the impossibility facts given in Tables~\ref{tab:window-adversaries} and~\ref{tab:leaky-bucket-adversaries}, we show other ones.
We prove that no algorithm that withholds the channel can be stable for even three stations against leaky-bucket adversaries.
The other impossibility facts hold for algorithms we call ``retaining'' and ``queue-size oblivious,'' they are defined in Section~\ref{sec:technical}.
We show that no retaining algorithm can be stable in a system of at least four stations against the window adversary of burstiness~$2$, which implies that no queue-size oblivious algorithm is stable in such environments.
By way of preparing a background for the main results of this paper, we show two preliminary facts: one is that centralized algorithms can achieve fair latency in systems of arbitrary size, the other is that fairness alone can be achieved by an acknowledgment based algorithm.

The possibility-type contributions of this paper hold for channels without collision detection while the impossibility results hold for channels with collision detection, unless indicated otherwise.


\begin{table}
\begin{center}
\begin{tabular}{|c ||c |c |}
\hline
\RB \LB
& $\mathsf{n=1}$ & $\mathsf{n\ge 2}$  \\
\hline\hline
\RB \LB
\textsf{ack based} & fair latency : possible & \\
\hline
\RB \LB
\textsf{full sensing} && stable : impossible\\
\hline
\RB
\textsf{\raisebox{-1.9ex}{general}} && stable : possible\\
&& \textrm{\raisebox{0.9ex}{stable and fair : impossible}}\\
\hline
\end{tabular}
\parbox{\pagewidth}{
~
\caption{\label{tab:leaky-bucket-adversaries}
Leaky-bucket adversaries: some possibility and impossibility facts regarding the quality of service of algorithms,  depending on a number $n$ of stations and a subclass of algorithms.
}}
\end{center}
\end{table}


\Paragraph{Structure of the document.}

This paper is organized as follows.
We review the methodology and technical notions in Section~\ref{sec:technical}.
The remaining three sections correspond to sizes of systems.
The smallest system of two stations is considered in Section~\ref{sec:two}.
The system of three stations is discussed in Section~\ref{sec:three}.
Finally, arbitrarily large systems are considered in Section~\ref{sec:four}.
We conclude in Section~\ref{sec:conclusion}.

\section{Technical Preliminaries}

\label{sec:technical}

A multiple access channel is a broadcast system with specific properties that we discuss in this section.
We also define adversarial models, broadcast algorithms and their subclasses, and measures of quality of service.
We use the letter~$n$ to denote the number of stations attached to a communication medium.
Each of the $n$ stations has a unique name assigned to it.
We assume that each station's name is an integer in the range $[1,n]$.
We also assume that every station knows~$n$, in the sense that $n$ can be a part of code of  a broadcast algorithm.
We use the letters $p$, $q$, $r$, and~$s$ for stations.
More precisely, when there are at least two stations in a system, then some two of them are denoted by~$p$ and~$q$; a third station, if available, is referred to as~$r$; and a fourth one, if available, if denoted by~$s$.


\Paragraph{Multiple access channel.}

What makes a broadcast system \emph{multiple access channel} is the property that a transmission by a station is instantaneously and successfully received by all the stations if and only if the transmission does not overlap with transmissions by other stations.
A message successfully broadcast is said to be \emph{heard} on the channel.

We consider a synchronous channel in which executions of algorithms are structured as sequences of events occurring at consecutive rounds so that overlapping transmission occur at the same round.
Each station is equipped with a clock.
The clocks are synchronized so that clock cycles begin simultaneously and are of the same duration.
A \emph{round} is defined to be the minimum number of clock cycles needed to transmit a message, with the local computation at a station considered to be of negligible duration.
The length of a round is the same across all the stations.
This allows to consider time as ``slotted'' into rounds so that every station  performs a transmission to fit in a round.

Multiple transmissions at the same round result in conflict for access to the channel, which is called a \emph{collision}.
When no stations transmit at a round, then the feedback that the stations receive from the channel is called \emph{silence}; we may also say about such a round that the channel or the round \emph{is silent}.
A channel may be equipped with a \emph{collision detection} mechanism, which makes the stations able to distinguish between silence and collision at a round.
If no collision detection mechanism is available, then stations perceive collisions as silences.


\Paragraph{Adversaries.}

We consider the worst-case performance of algorithms that handle traffic determined by an adversarial setting.
The \emph{rate of packet injection} of an adversary means an upper bound on the average number of packets injected into all the stations.
The maximum number of packets that an adversary may inject into all the stations at a round is called the  \emph{burstiness}.
An adversary is defined by these two parameters: injection rate and burstiness.
An injection rate can be defined in various ways depending on the class of time segments  over which we average.
We consider two kinds of adversaries: window adversaries and leaky-bucket ones.
In the context of adversarial queuing, window adversaries were first used by Borodin, et al.~\cite{BorodinKRSW-JACM01} and leaky-bucket adversaries by Andrews, et al.~\cite{AndrewsAFLLK-JACM01}.

We will use the letter $\rho$ to denote injection rate; we require the inequalities $0<\rho\le 1$ to hold.
Let $\rho$ be an injection rate and $w$ a positive integer:
the \emph{window adversary of type $(\rho,w)$} may inject at most $\rho w$ packets in each contiguous segment of $w$ rounds into any set of stations; the number~$w$ is called the \emph{window size} in such a context.
Let $\rho$ be an injection rate and $b$ a non-negative integer:
the \emph{leaky-bucket adversary of type $(\rho,b)$} may inject at most $\rho t+b$ packets in every contiguous segment of $t>0$ rounds into any set of stations.
An adversary is said to be \emph{of injection rate $\rho$} when it is either of window type $(\rho,w)$ or of leaky-bucket-type $(\rho,b)$, for some $w$ and~$b$, respectively.
The window adversary of type $(\rho,w)$ has burstiness $\lfloor \rho w\rfloor$.
The leaky-bucket adversary of type $(\rho,b)$ has burstiness $\lfloor \rho+b\rfloor$.

Observe that injection rate larger than~$1$ would allow the adversary to make the number of packets queued at stations grow unbounded, as at most one packet can be heard  per round.
In this paper, we consider only adversaries of injection rates exactly~$1$.
Such adversaries differ among themselves by their burstiness, which is either the window size~$w$ for a window adversary, or the number $b+1$ for a leaky-bucket adversary.
The models of window adversaries and leaky-bucket ones are equivalent for injection rates strictly less than~$1$, while the leaky-bucket adversary of injection rate~$1$ can generate sequences of injections not captured by any window adversary of rate~$1$, as was showed by Ros\' en~\cite{Rosen-IPL02}.
It follows that a possibility-type of a result for leaky-bucket adversaries holds automatically for window adversaries, while an impossibility result showed for window adversaries holds automatically for leaky-bucket adversaries.


\Paragraph{Distributed deterministic broadcasting.}

We consider broadcasting algorithms for a system of stations attached to a multiple access channel.
The algorithms are restricted to be both distributed and deterministic.
Further requirements for broadcast algorithms may stem from additional assumptions about channels.
For instance, when a channel is equipped with a collision detection mechanism then algorithms need to be able to react differently to collisions as opposed to silent rounds.
The ``slotted'' model of full synchrony does not impose any restrictions on specifications of algorithms: stations simply act driven by their local clocks.
On the other hand, the quality of broadcasting, for instance as measured by the maximum throughput achievable for a given number of stations, may be affected by the level of synchrony among the clocks at stations.

Any definition of algorithms needs to reflect general expectations of what algorithms are to accomplish.
We simply want algorithms to facilitate hearing packets on a channel, in a distributed and deterministic manner, while packets are injected into stations.
There are numerous tacit assumptions in such expectations, let us next discuss the key ones.

We want a station to have private memory to store incoming packets while waiting for access to the channel, as there may be a contention for access while packets are injected.
Such memory is called a \emph{buffer} or a \emph{queue}.
The following operations on a queue are always available: enqueuing a packet and dequeuing the queue.
An algorithm may use other operations on the queues, for instance to verify if the queue is empty or obtain the number of packets in the queue.
As the name ``queue'' suggests, a queuing discipline is used to prioritize packets in the buffer.
Queuing disciplines do not affect some measures of quality of service, like the maximum number of packets in queue, but may affect others, for instance, the time spent by packets in queues.
The first-in-first-out (FIFO) queuing organization appears to be most natural to optimize for packet latency; we assume that algorithms use FIFO queues to store packets.

The  packet a station is processing to have it transmitted first among all its packets is called \emph{pending}.
Once a station starts processing a packet we want the station to continue working on this very packet until a successful transmission, rather than possibly abandon the packet temporarily and switch to some other packet, as if expecting luck to be associated with packets.
We do not want a pending packet to be discarded without a prior successful transmission and we want each packet to be heard on the channel precisely once rather than multiple times.
After a pending packet has been heard, we do not care what happens to it, but it is natural to expect that such a packet is discarded to economize on space.
When a pending packet is discarded and the queue is non-empty, the next pending packet is obtained by dequeuing the queue.
It is possible for a station to have a pending packet and an empty queue, as the pending packet was obtained by dequeuing the queue so it is stored outside of the queue.
All of the above expectations have been incorporated in the definition of deterministic distributed broadcast algorithms used in this paper.


\Paragraph{Algorithms as automata.}

It is advantageous to formally model broadcast algorithms as automata, in a way that is now standard in representing distributed systems~\cite{Lynch-book96}.
The notion of automaton is especially helpful in formalizing impossibility proofs.
Our presentation of algorithms is above the level of automaton specification: we describe algorithms precisely enough to convince a reader that a full specification as an automaton is a matter of filling in details.
 
A \emph{state} of a station executing an algorithm is determined by the values of its private variables specified in a code of the algorithm and the number of packets that have been injected and still need to be transmitted.
One read-only variable is reserved to store the name of a station, while another one is used to store the number~$n$ of stations in the system.
There is an \emph{initial state} for each station in which the station starts an execution, the queue is empty at such a state.

A \emph{message} sent to the channel consists of a transmitted packet and possibly additional control bits.
The contents of transmitted packets do not affect state transitions, in the sense that every packet is treated as an abstract token, while control bits may affect state transitions.
A round during which no packet is heard as transmitted on the channel is said to be \emph{void}.
It is merely either a silence or a detection of collision or a message consisting of only control bits that is obtained by a station from the channel during a void round.

An \emph{execution} of an algorithm is a sequence of events occurring at consecutive rounds.
An execution starts with all the stations in their initial states.

An \emph{event} at a round is represented by the following sequence of actions at each station:
\begin{enumerate}

\item[(i)]  The station either transmits a message or pauses, accordingly to its state.

\item[(ii)] The station receives a feedback from the channel in the form of either hearing a message or collision signal or silence.

\item[(iii)]  New packets are injected into the station. 

\item[(iv)] 
A state transition occurs at the station.
\end{enumerate}

A state transition referred to at~(iv) above is based on the state at the end of the previous round, the feedback from the channel at this round, and the packets injected at this round.
It includes the following operations.
If new packets have been injected at this round then they are all enqueued.
If the station has just transmitted successfully at this round, then the transmitted packet is discarded and new pending packet is obtained by dequeuing the queue, which produces a packet that has been waiting longest in the queue, unless the queue is empty.
A message for the next round is prepared, if any will be attempted to be transmitted.

Such a general representation of algorithms as automata requires stations to listen to the channel at all rounds.
Some weak algorithms we consider, namely acknowledgment based ones, can be interpreted as having a station ``ignore'' the feedback from the channel when the station has no packets to transmit.
The modeling of algorithms by automata we assume does not exclude such algorithms, because ``ignoring'' the feedback from the channel can be represented by cycling in the same state.
Listening to the channel at all rounds has often been referred in the literature as ``full sensing'' to indicate that stations without packets do not ignore the feedback from the channel and undergo state transitions at every round.


\Paragraph{Properties of algorithms.}

When a general broadcast algorithm is executed, a feedback received by the stations from the channel at a round is of the following five kinds: (1) packet with control bits piggybacked on it, (2) packet without control bits attached to it,  (3) control bits without any packet, (4) silence, or (5) collision.
An algorithm designed to have stations actually send control bits in some messages is called \emph{adaptive}.
This term is to express the property that stations may adapt their behavior  to ``instructions''  encoded by control bits in a message.
For instance, a station may announce that it is reserving the channel for a sequence of consecutive transmissions, which is to make other stations postpone any transmissions and instead merely wait until the reservation is recalled.

We consider two subclasses of general algorithms: full sensing algorithms and acknowledgment based ones.
A broadcasting algorithm is \emph{full sensing} if no control bits are used in messages, neither attached to packets nor transmitted as separate messages.
For such algorithms, a feedback received from the channel is of the following three forms: (1) packet, (2) silence, or (3) collision.
By the definitions of adaptive and full sensing algorithms, an algorithm is adaptive when it is not full sensing, and any algorithm is either adaptive or full sensing.

A broadcasting algorithm is defined to be \emph{acknowledgment based} when a state transition at a station depends only on which consecutive round it is being spent to process a  currently handled packet, where the numbering of rounds starts from the first round when the packet started to be processed.
A feedback received from the channel, that matters for a station running an acknowledgment based algorithm, is of two possible forms: (1) the packet just transmitted by the station or (2) anything else.
Once a packet transmitted by a station is heard on the channel, then this is considered to be ``acknowledgment'' of the broadcast, which results in the station starting to process the next packet.

Acknowledgment based algorithms are full sensing in that they do not use control bits in messages.
A station~$p$ running an acknowledgment-based algorithm is oblivious to the actions of other stations except for the rounds when $p$ transmits.
An acknowledgment based algorithm can be considered as determined completely by an unbounded binary sequence assigned to each station, which we call the \emph{transmission sequence} of the station; a separate sequence assigned for each station.
A transmission sequence is to be interpreted as follows: if the $i$th entry is a~$1$, then the station transmits the currently processed packet at the $i$th round, counting from the first round when the packet was started to be processed, and a~$0$ means that the station pauses at the $i$th round of processing the packet.

We show that the three categories of algorithms just introduced are different in terms of their capability to achieve  desirable properties by algorithms.
This occurs already in systems of small sizes.
In particular, we show that no acknowledgment based algorithm can be stable in a system of two stations, while there is a full sensing algorithm of  fair latency against window adversaries for such systems; see Section~\ref{sec:two}.
Similarly, we show  that no full sensing algorithm can be stable in a system of three stations, while there is an adaptive algorithm of fair latency against window adversaries for such systems; see Section~\ref{sec:three}.

A natural paradigm to organize a broadcast algorithm is for stations to be greedy by withholding the channel: once a station~$p$ transmits successfully at a round, then $p$ keeps transmitting as long as there are  packets waiting in~$p$.
A broadcast algorithm \emph{withholds the channel} when stations behave in this  way in the course of an execution of the algorithm.

Given an algorithm, we say that \emph{the decision if to transmit the currently pending packet at the next round does not depend on the size of the queue at a round} when the following holds for any station running the algorithm: whenever the station has a pending packet and the queue at the station is nonempty and the pending packet is to be transmitted at the next round, then if a packet were removed from the queue at the end of the previous round then the pending packet would be transmitted at the next round nonetheless.
An  algorithm~$\cP$ is \emph{queue-size oblivious} if it has the following two properties:

\begin{enumerate}
\item[(1)] 
the decision if to transmit the current pending packet at the next round does not depend on the size of the queue at a round;
\item[(2)] 
if an execution of~$\cP$ is stable, for a sequence of injections of packets performed by some adversary, then the execution for the same injections is stable if allowed to be additionally disturbed by the adversary by removing a packet already in a queue at a round to inject it at the next round, a packet per round. 
\end{enumerate}

A definition capturing the property that an algorithm is indifferent to the sizes of queues needs to be somewhat technical.
This is because of the model of algorithms we work with: a station enqueues and dequeues packets itself, so the tally of the number of packets in the queue is implicitly known by a station, being the difference between the numbers of enqueued packets and dequeued ones.
The property of being queue-size oblivious is not related to adaptability of algorithms.
Suppose a system is running a queue-size oblivious algorithm: when a station dequeues the queue to obtain a new pending packet and then computes the number~$x>0$ of packets remaining in the queue and finally decides to transmit a message, then the message \emph{may} include control bits representing this number~$x$, but the decision to transmit a message with the pending packet at the next round would be the same if the number of packets in the queue were~$x-1$ and this were achieved by removing a packet from the queue at the end of the previous round.
Any acknowledgment based algorithm is queue-size oblivious.
All the algorithms given in~\cite{ChlebusKR-TALG12} are queue-size oblivious; these algorithms were analyzed  in~\cite{ChlebusKR-TALG12} for injection rates less than~$1$.
We call an algorithm \emph{queue-size sensitive} if it is not queue-size oblivious.
Among the algorithms described in this paper, some are queue-size oblivious, in particular  such is algorithm \textsc{2-Adaptive}, while others are queue-size sensitive,  for instance \textsc{2-Full-Sensing}, \textsc{3-Adaptive} and \textsc{Move-Big-To-Front}.


\Paragraph{The quality of broadcasting.}

The performance measures we consider express the quality of broadcasting.
Fairness and stability are the most general properties we want to achieve.
An algorithm~$\cP$ is said to be \emph{fair} against an adversary~$\cA$ if, for any execution of algorithm~$\cP$ against this adversary, every packet injected by $\cA$ is eventually heard on the channel.
An algorithm~$\cP$ is said to be \emph{stable} against an adversary~$\cA$  if, for any execution of algorithm~$\cP$ in a system of $n$ stations against this adversary, there is a number $s(n)$ such that the number of packets stored in the queues is at most $s(n)$ at any round.

As we show in Section~\ref{sec:two}, achieving both stability \emph{and} fairness against leaky-bucket adversaries is impossible, except for a trivial case of a system of a single station.
Achieving fairness alone against any adversary is straightforward in a full sensing way:
\textsc{Round-Robin} algorithm, which makes station~$i$ transmit at any round~$k$ such that $k \equiv i \pmod n$, does the job.
This is because when a packet is injected into a station, then it is of some rank, say, $j$ in the queue: during the following $j$th cycle of the execution the station hosting the packet transmits the packet successfully.


\begin{proposition}
\label{prop:fairness-ack-based}

Fairness can be achieved by an acknowledgment based algorithm against leaky-bucket adversaries.
\end{proposition}

\begin{proof}
We may assume that $n> 20$, since for  $n\le 20$ stations  we can use the construction for $21$~stations in which $21-n$ dummy stations have no packets injected.
A station~$p$ uses the $p$th prime number $x_p$ from the range $[n\ln n,3n\ln n)$.
There are at least $n$ primes in this range, by properties of the distribution of primes in intervals~\cite{Ribenboim-book1996}, therefore numbers $x_p$ are well defined.
A transmission sequence of station~$p$ has an occurrence of~$1$ at every position $3n^2\ln n +i\cdot x_p$, for a non-negative integer~$i$, and the remaining positions are all~$0$'s.

We argue that each station~$p$ with a pending packet at round~$t$ transmits successfully in the interval $[t,t+6n^2\ln n)$ against any leaky-bucket adversary.
Suppose that this is not the case, to arrive at a contradiction.
Consider the interval $\cI=[t+3n^2\ln n,t+6n^2\ln n)$, which consists of $3n^2\ln n$ rounds.
Station~$p$ transmits at least $\frac{3n^2\ln n}{x_p}\ge n$ times  in~$\cI$, exactly once in a  subinterval of $x_p$ rounds.

It is sufficient to show that any other station can transmit simultaneously with $p$
at most once in~$\cI$, since $n-1$ other stations compete with~$p$ for a slot to transmit.
To this end, consider a station~$q$ distinct from $p$ and assume that $q$ transmits simultaneously with~$p$ at least twice during~$\cI$.
Note that these transmissions must happen while station~$q$ processes the same packet.
This is because when a station starts working on a new packet, it does not perform a transmission in the first $3n^2\ln n$ rounds of this operation.
Therefore these two simultaneous transmissions of $p$ and $q$ occurring at some rounds $t_1<t_2$ must have the property that $t_2-t_1$ is divisible by both~$x_p$ and~$x_q$.
This however may happen only if $t_2-t_1$ is divisible by $x_p\cdot x_q$, because the numbers $x_p$ and~$x_q$ are both primes: thus the inequalities $t_2-t_1\ge x_p\cdot x_q \ge n^2\ln^2 n$ hold.
On the other hand, $3 n^2 \ln n \ge t_2-t_1$ and hence we obtain the inequality $3n^2\ln n \ge n^2\ln^2 n$, which is false for $n>20$. 
This means that  the points $t_1$ and $t_2$ cannot be both in~$\cI$, because otherwise we obtain a contradiction.
\end{proof}

An algorithm~$\cP$ is said to have \emph{bounded latency} against an adversary~$\cA$ if, in any execution of the algorithm in a system of $n$ stations against adversary~$\cA$, there is a number $\ell(n) >0$ such that, for each injected packet, the time interval from the injection of the packet until the packet is heard on the channel is of length at most~$\ell(n)$.
If a station stores some $x$ packets in its queue at a round of an execution of an algorithm, then the latency of the algorithm is at least~$x$ in this execution.
It follows that a bounded-latency algorithm is both fair and stable.
An algorithm of bounded latency, for a given adversary of injection rate~$1$, has \emph{fair latency} if the bound $\ell(n)$ satisfies the estimate $\ell(n)=\cO(\text{burstiness})$.
An algorithm of fair latency is also \emph{strongly stable}, in the sense that the number of packets queue at stations is $\cO(\text{burstiness})$.
Any distributed algorithm given in this paper, that is both fair and stable against some adversary, happens also to be of fair latency against this adversary.

There is a fundamental difference in quality of service between deterministic broadcast schedules implemented by distributed algorithms versus centralized ones: distributed algorithms cannot achieve both stability and fairness while centralized ones can achieve fair latency.
By a \emph{centralized} algorithm we mean one controlled by a central processing unit (CPU) that can communicate with the stations instantaneously to both collect information from them about injected packets and give them instructions about which station is to transmit at a given round.


\begin{proposition}
\label{prop:centralized-algorithm}

Fair latency can be achieved by a centralized algorithm against leaky-bucket adversaries in a system of $n$ stations, for any $n>0$.
\end{proposition}

\begin{proof}
Take a system of $n$ stations and a leaky bucket adversary of type $(1,b)$.
The adversary injects packets into stations subject to the constraints of its type.
Each station maintains its private FIFO queue.
An injected packet is immediately enqueued into the respective private queue.
When a station is to transmit at a round, it dequeues a packet from the private queue and transmits the packet.
Each station keeps a record of the number of packets injected at a round through the next round.

We specify an algorithm by a sequence of actions that occur at a round.
The CPU controls which station transmits at a round by notifying the designated station at the immediately preceding round.
The CPU maintains a FIFO queue to store tokens.
Each token is labelled with a name of station.

A round begins with a station designated to transmit at this round performing a transmission, if there is such a station.
After the stations received a transmission or detected a lack thereof, each station~$p$ notifies the CPU about the number $k_p$ of packets that were injected into~$p$ at the previous round.
The CPU creates $k_p$ tokens each labelled with $p$ and enqueues them in arbitrary order.
Next the CPU dequeues a token from the queue, unless the queue is empty.
If the queue is nonempty and a token~$g$ has been obtained by dequeuing the queue, let $q$ be the name of the station with which the token~$g$ is labelled.
The CPU notifies the station~$q$ to transmit at the next round.
This concludes the round.
Observe the invariant that when a station~$r$ is notified by the CPU to transmit, then the private queue of packets at~$r$ is nonempty.

Next we estimate the maximum packet delay.
We consider the round of injection and the round of transmission as not contributing to delay.
Observe that when the adversary injects a packet at round~$t$, then the station notifies the CPU at round~$t+1$ about the packet and at the same round the CPU notifies some station to transmit at round~$t+2$: this incurs a delay of at least~$1$ for the packet and the round~$t+2$ will not be silent.

We claim that a packet cannot be delayed by more than $b+1$~rounds.
Suppose that this is not the case, to arrive at a contradiction.
Take the first round~$t_1$ in which a packet~$u$ is transmitted that was delayed by at least $b+2$~rounds.
Let $t_0$ be the last silent round before round~$t_1$.
Round~$t_0$ exists, as the first round in any execution is silent.
The queue at the CPU was empty at round~$t_0-1$, so all the packets injected into stations by round~$t_0-2$ have been heard on the channel.
The adversary could not inject any packet at round~$t_0-2$, as otherwise the queue at the CPU could not be empty at round~$t_0-1$.
Packet~$u$ was delayed by more than~$b+1$ rounds so it was injected by round~$t_1-b-3$.
Therefore packet~$u$ was injected in the time interval~$\cI$ starting from round~$t_0-1$ through round~$t_1-3-b$.

Now observe two facts.
One is that the adversary could inject at most $(t_1-3-b) - (t_0-1) + 1+b=t_1-t_0-1$ packets in this interval~$\cI$.
The other is that the number of packets heard on the channel starting from round~$t_0+1$ through round~$t_1$ is $t_1-(t_0+1)+1=t_1-t_0$.
These two facts contradict one another.
\end{proof}

\section{Just Two Stations}

\label{sec:two}

In this section we consider two stations attached to a multiple access channel.
More precisely, the algorithms presented in this section are developed for a system of \emph{exactly two} stations, while impossibility results hold for any system with \emph{at least two} stations.
We consider both window and leaky-bucket adversaries.
For leaky-bucket adversaries, there is an algorithm that can handle traffic in a stable way, while achieving both stability \emph{and} fairness is impossible.
The situation is different with respect to window adversaries: there exists a full sensing fair-latency algorithm, which is a fortiori stable and fair, while no acknowledgment based algorithm can merely be stable.


\Paragraph{Leaky-bucket adversaries for two stations.}

We begin by showing that stability can be achieved for two stations.
Consider the following queue-size oblivious algorithm that we call \textsc{2-Adaptive}.
The two stations $p$ and~$q$ use a conceptual ``token.'' 
The initial state is such that the queues are empty and the token is with~$p$.
The token gives the privilege to transmit packets: a station with the token keeps transmitting for as long as it has a pending packet.
The token eliminates a possibility of collisions.
When a station with a token has only one packet at the beginning of a round, which means it is a pending packet while the queue is empty, then the message to carry the pending packet contains an ``over'' bit.
The other station takes over after receiving the ``over'' bit in the message and considers itself to be holding the token.
When a station receives the token at round~$i$ and does not have a pending packet,  then the station does not transmit at the next round~$i+1$.
A station, say~$p$, that does not hold a token at a round does not transmit at this round, but when this round is silent, then~$p$ considers itself holding the token at the next round.
This allows to exchange the token without any transmissions when no station has a pending packet.
Next we argue why algorithm  \textsc{2-Adaptive} is stable but not fair in a system of two stations against leaky-bucket adversaries.

Consider an adversary of some burstiness $b+1$ and a segment of contiguous transmissions in an execution.
For this segment, the sum of the sizes of the queues is not more than the number of packets during the first transmission of this segment plus burstiness, because a packet per round is transmitted.
We may assume without loss of generality that the adversary exercises its power to use burstiness as early as possible for the considered segment of contiguous transmissions.
We examine the number of packets in the queues at a round when a packet is transmitted by referring to what happened at the most recent preceding silent round.
Such a silent round exists, as the first round in any execution is silent.

For any round~$t$ at which some station transmits a packet, a transmission of a message will occur at the next round~$t+1$ unless two things occur: the transmitting station has only one packet at the beginning of round~$t$, and the other station does not have any packet at the end of round~$t$.
Therefore a silent round~$t+1$ after a transmission at round~$t$ indicates that at the point in time that is simultaneously the end of round~$t$ and the beginning of round~$t+1$ the number of packets in the stations equals at most the burstiness of the adversary, as these many packets could be injected at round~$t$ into the station that transmitted its solitary packet at round~$t$ while the other station did not have any packets to transmit at round~$t+1$.
Such a scenario leads to up to $b+1$ packets in the queues.
Consider now a scenario in which a silent round~$t$ is not preceded by a round with transmission.
Suppose station~$p$ is scheduled to transmit at round~$t$.
The adversary can inject up to $b+2$ packets in rounds $t$ and~$t+1$.
There is a possibility that these many packets are simultaneously in the queues: this occurs when $b+1$ packets are injected into~$p$ at round~$t$ and next one packet again at round~$t+1$.
These two cases show that there are never more than $b+2$ packets in the queues against this adversary, which means that the algorithm is stable.

To show that algorithm  \textsc{2-Adaptive} is not fair if $b>0$, consider the following scenario.
Let the adversary of burstiness~$2$ keep injecting one packet per round into station~$q$, starting from the first round, while initially no packets are injected into~$p$.
Station~$p$ does not transmit at the first round as it does not have any packet at the beginning of the round.
This results in $q$ obtaining the token.
Station~$q$ prepares a message at the end of the first round to be transmitted at the second round: the message contains the pending packet and the ``over'' bit as the queue at~$q$ is empty.
The message gets transmitted at the second round and station~$p$ obtains the token.
Station~$p$ does not transmit at round three, as it has no packets, so station~$q$ receives  the token back.
At the end of round three, station~$q$ has two packets: those injected at rounds two and three.
Hence the message prepared by $q$ at the end of round three contains a packet but not the ``over'' bit, as the queue still contains one packet.
The message gets transmitted at round four and a new packet is injected into~$q$ at this round.

The transmissions of~$q$ will continue forever, as the queue at~$q$ will always contain at least one packet.
Let the adversary inject one packet into~$p$ at round four.
This is consistent with burstiness~$2$.
This packet will never be transmitted, as station~$p$ will never receive the token, which means that the algorithm is not fair.

Algorithm \textsc{2-Adaptive} withholds the channel and is queue-size oblivious.
We will show that no algorithm that withholds the channel can be stable for a system of three stations, see Theorem~\ref{thm:impossible-3-stations-withhold-stable}, and that no algorithm that is queue-size oblivious can be stable in a system of four stations, see Corollary~\ref{cor:impossible-queue-size-oblivious-stable}.
The following question is natural to ask now: is there an algorithm with a bounded packet latency for a system of two stations?
The answer depends on the kind of adversary.
We answer this question in the negative for leaky-bucket adversaries:


\begin{theorem}
\label{thm:2-impossible-fair-stable}

No algorithm can be both stable and fair for a system of at least two stations against the leaky-bucket adversary with burstiness~$2$.
\end{theorem}

\begin{proof}
Consider a system with at least two stations~$p$ and~$q$.
Suppose, to arrive at a contradiction, that some algorithm~$\cP$ for this system is both stable and fair.
We construct an execution~$\cE$ of~$\cP$ that contains infinitely many void rounds while simultaneously  the adversary injects one packet per round on the average.
Once such an execution~$\cE$ is shown to be possible, a contradiction with the stability of~$\cP$ will have been established.

We determine an execution~$\cE$  by specifying two components.
One is a sequence $t_{0}, t_{1},t_{2},\ldots$ of void \emph{milestone rounds}, where $t_{i}<t_{i+1}$ for any $i\ge 0$.
The other is how packets are injected to enforce the occurrence of milestone rounds.
The construction is by induction on round numbers.
Define $t_{0}$ to be the first round of the execution.
Let the adversary inject one packet into~$p$ at the first round.
Round~$t_0$ is void in~$\cE$ because no packets have been injected before the round.
Suppose the execution~$\cE$ and injections have been defined by a milestone round~$t_{i}$, for some $i\ge 0$.
Next we define the milestone round~$t_{i+1}>t_{i}$ and how packets are injected in the interval starting from~$t_{i}+1$ through~$t_{i+1}$ in~$\cE$.

Consider a conceptual scenario~$\cS_1$ in which the adversary keeps injecting one packet per round into~$p$, starting from round~$t_{i}+1$, such injections occurring forever.
If there is a void round~$x>t_i$ under $\cS_1$ then we extend $\cE$ to round~$x$ as follows.
Set $t_{i+1}=x$ and let the adversary inject one packet into~$p$ at any round after $t_i$ and up to round~$t_{i+1}=x$.
This behavior of the adversary in scenario~$\cS_1$ is consistent with type $(1,1)$ of adversary, in that exactly one packet per round is injected in the newly added segment of the execution.
Otherwise, when $x$ does not exist under scenario~$\cS_1$, then this means that a packet is heard at every round after~$t_i$.
Since the algorithm is fair, eventually station~$q$ will have an empty queue; let $y$ be the first such a round after round~$t_{i}$.
Observe that $p$ will transmit a packet at every round after $y$ under scenario~$\cS_1$: there are no packets in~$q$ and a packet is heard at every round.

Consider next a scenario~$\cS_2$ that differs from~$\cS_1$ only in that the adversary injects a packet into~$q$ at round~$y$.
Observe that $p$ will transmit at every round after $y$ under scenario~$\cS_2$, because these two scenarios do not differ with respect to injections into~$p$.
Under scenario~$\cS_2$, station~$q$ has a pending packet starting from round~$y+1$.
Station~$q$ needs to transmit the pending packet at least once, because otherwise the algorithm would not be fair under scenario~$\cS_2$.
Let $q$ transmit this packet for the first time at some round~$z>y$.
Since $p$ also transmits at~$z$, a collision occurs, which makes round~$z$ void.
We extend $\cE$ as follows.
Define $t_{i+1}$ to be this round~$z$.
Let the adversary inject packets according to scenario~$\cS_2$ starting from the next round after~$t_i$ and ending just before round~$t_{i+1}=z$, and no packets at round~$t_{i+1}=z$.

Continuing this construction defines an execution~$\cE$.
The adversary injects either no packet or one packet or two ones at a round of~$\cE$.
When the rounds of one packet injected are removed from the execution, then what remains are alternating rounds with respect to the number of injected packets: two packets injected followed immediately by no packets injected, and no packets followed immediately by two packets.
This pattern of injections in~$\cE$ is consistent with the definition of the leaky-bucket adversary of burstiness~$2$.

All the milestone rounds~$t_i$ are void.
Moreover, the average number of packets injected between the rounds $t_i$ and $t_{i+1}$ is precisely one packet per round, by the specification of the behavior of the adversary.
It follows that the number of packets at the queues of the two stations at round~$t_{i}$ is at least~$i+1$, for $i \ge 0$, and thus grows unbounded.
The existence of execution~$\cE$ contradicts the stability of algorithm~$\cP$.
\end{proof}

The impossibility of Theorem~\ref{thm:2-impossible-fair-stable}, which is about deterministic distributed algorithms, could be compared to the possibility of Proposition~\ref{prop:centralized-algorithm}.
We pause now to revisit the definition of an algorithm we work with to capture the notion of a ``broadcast schedule'' for multiple access channels.
For the proof of Theorem~\ref{thm:2-impossible-fair-stable} to be valid, the following property of such broadcast schedules needs to be satisfied: for any round~$t$, the action of each station~$p$ at this round is uniquely determined by what has been heard on the channel prior to round~$t$ and how the adversary has been injecting packets into~$p$ prior to round~$t$.
In particular, if there are two scenarios $\cS_1$ and~$\cS_2$ such that the behavior of an adversary is the same prior to some round~$t$ in both $\cS_1$ and~$\cS_2$ but the behavior differs at round~$t$ in precisely that (1) the adversary injects the same number of packets into a station~$p$ at~$t$ in both $\cS_1$ and~$\cS_2$ while (2) the numbers of packets injected into some other station~$q$ differ in the scenarios~$\cS_1$ and~$\cS_2$, then the actions taken by the station~$p$ at round~$t+1$ are the same under both scenarios, while only the actions of~$q$ may differ.
The property that the actions of stations at a round are uniquely determined by the history prior to the round follows from the assumption that algorithms are deterministic.
The property that the actions of stations at a round are affected by other stations only through what has been heard on the channel prior to the round holds because algorithms are assumed to be distributed.

The sequence $t_{0}, t_{1},t_{2},\ldots$ of milestone rounds, as considered in the proof of Theorem~\ref{thm:2-impossible-fair-stable}, may have the property that there is no upper bound on the differences $t_{i+1}-t_{i}$, for $i\ge 0$.
If such an upper bound, say $w$, existed, then the behavior of the adversary would be  consistent with the definition of type $(1,w)$ window adversary.
It follows from Theorem~\ref{thm:2-full-sensing-bounded-latency}, given in the part of this section devoted to window adversaries, that one cannot strengthen the proof above to show existence of such a bound, as otherwise impossibility would hold for window adversaries.
Regarding window adversaries, possibility is shown for three stations in Theorem~\ref{thm:3-adaptive-fair-latency} of Section~\ref{sec:three}, while impossibility holds for four stations, as shown in Corollary~\ref{cor:impossible-fair-stable-4-stations-window} of Section~\ref{sec:four}.

When algorithms are restricted to be full sensing, then even stability is not possible to be achieved  in a system of two stations, as we show next.
The proof of Theorem~\ref{thm:impossible-2-full-sensing-stable} relies on the property that messages cannot carry control bits to redirect the course of action of listening stations.
This property can be formally abstracted as follows:
Consider a full sensing algorithm and two scenarios $\cS_1$ and $\cS_2$ for a system of two stations.
If  station~$p$ has the same injections up to round~$t$ under the two scenarios, while station~$q$ transmits at the same rounds up to round~$t$ under the two scenarios, then $p$ cannot detect any difference between the injections into~$q$ under the two scenarios up to round~$t$, so $p$ is at the same state at the beginning of round~$t$ under both $\cS_1$ and~$\cS_2$.


\begin{theorem}
\label{thm:impossible-2-full-sensing-stable}

No full sensing algorithm can be stable for a system of at least two stations
against the leaky-bucket adversary with burstiness~$2$.
\end{theorem}

\begin{proof}
We show that the adversary can enforce an execution $\cE$ of $\cP$ with infinitely many void rounds, while maintaining the average injection rate~$1$.
Similarly as in the proof of  Theorem~\ref{thm:2-impossible-fair-stable}, we identify consecutive milestone rounds $t_{0}, t_{1},t_{2},\ldots$, where $t_{i}<t_{i+1}$ for any $i\ge 0$, and specify how packets are injected at every round.
If a void round occurs before the end of the described construction, then this gives a milestone round and we are done.
Thus we will always consider a scenario in which the current round is not a void one, whenever logically possible.

Let the adversary inject one packet into~$q$ at the first round, and define $t_{0}$ to be the first round of execution~$\cE$ of~$\cP$.
Round~$t_0$ is void in~$\cE$ because no packets have been injected before the round.
Suppose the execution~$\cE$ and injections have been defined by a milestone round~$t_{i}$, for some $i\ge 0$, we define the milestone round~$t_{i+1}>t_{i}$ and how packets are injected in the interval starting from~$t_{i}+1$ through~$t_{i+1}$ in~$\cE$.

We consider a sequence of conceptual scenarios.
In the first scenario~$\cS_1$, the adversary continuously injects one packet per round into station~$q$, starting from round~$t_{i}+1$, such injections occurring forever.
If there is a void round~$x>t_i$, then we extend $\cE$ to round~$x$ as follows.
Set $t_{i+1}=x$ and let the adversary inject one packet into~$q$ at any round after $t_i$ and up to round~$t_{i+1}=x$.
Otherwise, when $x$ does not exist under scenario~$\cS_1$, then this means that a packet is heard at every round after~$t_i$.
Since the algorithm is stable, station~$p$ pauses at some round~$y_1 > t_{i}$ while station~$q$ successfully transmits at round~$y_1$.
Let the second scenario~$\cS_2$ differ from~$\cS_1$ in that the adversary injects additionally one packet into~$p$ at round~$t_i+1$.
Station~$q$ again transmits at round~$y_1$ in scenario~$\cS_2$, since there is no difference between the pattern of injections into this station.
Therefore  $p$ pauses at round~$y_1$ in scenario~$\cS_2$ to avoid collision.
The third scenario~$\cS_3$ is a modification of $\cS_2$ in only that the adversary does not inject a packet into~$q$ at round~$t_i+1$.
Again station~$p$ pauses at round~$y_1$ in scenario~$\cS_3$, since there is no difference between the injection pattern into~$p$ in scenarios~$\cS_2$ and $\cS_3$, and because $p$ pauses in~$\cS_2$.
Suppose station~$q$  transmits at round~$y_1$ in scenario~$\cS_3$, similarly as in~$\cS_2$, which is the case unless there are no packets in~$q$ at the round.
The forth scenario~$\cS_4$ is a modification of $\cS_3$ in only that the adversary injects a packet into~$p$ at round~$t_i+2$.
The fifth scenario~$\cS_5$ is a modification of $\cS_4$ in only that the adversary does not inject a packet into~$q$ at round~$t_i+2$.

We continue along these lines to define a sequence of scenarios.
The applied pattern is such that first an injection into~$p$ is added at a round~$t> t_i$ in a scenario, followed by a scenario in which an injection into $q$ at the same round~$t$ is omitted.
Next round~$t+1$ is processed in the same way, and so on.
Since the modification of injections occurs in only one station at a time, there is no difference between the new scenario and the immediately preceding one from the point of view of the other station that is not affected, so the pattern of transmissions is preserved.

Define such a sequence of scenarios that is sufficiently long to redefine injections up to round~$y_1$: in a new pattern packets are injected into $p$ only up to round~$y_1$, while afterwards the adversary injects only into $q$ and a packet is heard at every round.
Since the algorithm is stable, station~$p$ pauses at some round~$y_2> y_1$ while station~$q$ successfully transmits at round~$y_2$.
We continue to redefine injections up to round~$y_2$ in such a way that packets are injected only into $p$ up to round~$y_2$ while afterwards packets are injected only into~$q$.

This pattern is repeated through a sequence of rounds $y_i$, for $1\le i\le k$, in which $q$ transmits, while packets are injected only into $p$ up to round~$y_k$, for arbitrarily large integer~$k$.
Take $k$ larger than the number of packets in $q$ at round~$t_i$: this results in a contradiction since $q$ transmits more packets by round~$y_k$ then there are in its queue while no new packets have been injected into~$q$.
This means that a void round occurs by round~$y_k$.
Define $t_{i+1}$ to be the first such a void round after $t_i$ in the first scenario~$\cS$ in which a void round occurs after~$t_i$.
We extend $\cE$ beyond $t_i$ up to round~$t_{i+1}$ as follows.
The adversary injects packets exactly as in $\cS$ after $t_i$ but before $t_{i+1}$.
If the adversary injected two packets at some round in $\cS$ after $t_i$ but before $t_{i+1}$ then the adversary pauses at round~$t_{i+1}$, otherwise the adversary injects one packet into some station at round~$t_{i+1}$.

All the milestone rounds~$t_i$ are void while the average number of packets injected between~$t_i$ and~$t_{i+1}$ is precisely one packet per round, by the specification of the behavior of the adversary that is of burstiness~$2$.
It follows that the number of packets at the queues of the two stations at round~$t_{i}$ is at least~$i+1$, for $i \ge 0$.
The existence of an execution~$\cE$ in which the number of packets in queues grows unbounded means  instability of algorithm~$\cP$.
\end{proof}


\Paragraph{Window adversaries for two stations.}

Next we show that fair latency can be achieved in a system of two stations by a full sensing algorithm against any window adversary.
We start with a full sensing algorithm that has a positive integer number $i$ in its code  interpreted as a window; the algorithm is called \textsc{2-Full-Sensing($i$)}.

An execution of the algorithm is structured as a sequence of consecutive phases.
A \emph{phase} consists of exactly $i$~rounds.
Packets injected in the course of a phase are considered to be \emph{available} during the next phase.
A phase is used to broadcast precisely the packets injected in the immediately preceding phase and thus available during the phase.

All stations simply wait in the first phase.
Consider the first round of one of the next phases.
If one of the stations realizes that it contains exactly $i$ available packets, then the station  spends the whole phase of $i$ rounds transmitting these packets.
Suppose that none of the stations contains exactly $i$ available packets.
In this case $p$ starts transmitting at the first round of the phase.
Station~$q$ counts the number of  packets transmitted by~$p$ and starts transmitting packets immediately after silence or when $q$ realizes at a round that its number of available packets is equal to the number of the remaining rounds of the phase.
Each station needs to maintain two counters: one is the size of the queue and the other is the consecutive round in a phase.
The count of rounds is updated by incrementing the value by~$1$ modulo $i$, while the count of the queued packets is updated following insertions and successful transmissions.

Next we define an algorithm \textsc{2-Full-Sensing} to handle an adversary of arbitrary window.
If we knew the adversary and began an execution by invoking \textsc{2-Full-Sensing$(i)$} where the window~$i$ implicit in \textsc{2-Full-Sensing$(i)$} were the same as that of the adversary, then no collisions would occur, by the properties of \textsc{2-Full-Sensing$(i)$} discussed above.
As we do not know the adversary, the algorithm works by trying to run algorithms \textsc{2-Full-Sensing$(i)$} to test window sizes~$i$ for consecutive values of~$i$, starting from \textsc{2-Full-Sensing$(1)$} to test $w=1$.
A collision can possibly occur at an event: it can be detected by all the stations, since both of them transmit while none can hear a packet.
If a collision occurs in the course of executing \textsc{2-Full-Sensing$(i)$}, then both stations invoke \textsc{2-Full-Sensing$(i+1)$}.
The packets that are already in the queues when algorithm \textsc{2-Full-Sensing$(i+1)$} is invoked are called \emph{old}.
A transition from  \textsc{2-Full-Sensing$(i)$}  to algorithm \textsc{2-Full-Sensing$(i+1)$} starts by handling such old packets.
Consider the beginning of executing \textsc{2-Full-Sensing$(i+1)$} just after a transition has been made.
The stations~$p$ and~$q$ begin by taking care of their old packets: first station~$p$ transmits all its old packets, and after a silent round station~$q$ does the same.
While old packets are being unloaded, the actions of the adversary performed at consecutive rounds for a station are not immediately acted upon by the station but instead are enqueued in an additional private queue DELAY operating in a FIFO manner.
The queue DELAY stores an entry for each round of which a record of the adversary's actions  needs to be kept. 
An entry is a record of the number of packets injected at the respective round. 
The injected packets themselves are enqueued into the only queue at a station used to store waiting packets.
After the station~$q$ has completed transmitting its old packets, which is indicated by a silent round, both stations start executing instructions of \textsc{2-Full-Sensing$(i+1)$}.
To this end they would dequeue DELAY while simultaneously enqueueing the current actions of the adversary if needed.
Recall that an invocation of \textsc{2-Full-Sensing$(i+1)$} starts with the first phase of $i+1$ rounds during which the stations pause while storing packets available for the next phase.
When multiple invocations of \textsc{2-Full-Sensing$(k)$} are performed back-to-back, for consecutive values of~$k$, we may modify the action performed during a transition from $i$ to $i+1$ to save on void rounds. 
Namely, let the stations begin \textsc{2-Full-Sensing$(i+1)$} by dequeuing either $i+1$ entries from the queue DELAY or all of them, whichever is smaller, and let this contribute towards the beginning of the first phase of \textsc{2-Full-Sensing$(i+1)$}.
Afterwards they dequeue one entry per round, unless DELAY is empty, in which case it is not used.


\begin{theorem}
\label{thm:2-full-sensing-bounded-latency}

Algorithm \textsc{2-Full-Sensing} is full sensing and has fair latency in a system of two stations against any window adversary.
\end{theorem}

\begin{proof}
The property of being full sensing follows directly by examining the design of the algorithm.
For instance, station~$q$ may take over from $q$ within a phase after hearing silence,
which does not require the power of adaptive algorithms to send control bits.

Consider an adversary of some type  $(1,w)$.
We claim that the set of numbers~$i$ such that \textsc{2-Full-Sensing$(i)$} is invoked in an execution makes a bounded interval.
This follows from two facts.
One is that starting from some invocation of \textsc{2-Full-Sensing$(k)$}, for $k\le w$, the behavior of the adversary during the remaining part of the execution is consistent with the type $(1,k)$ adversary; let $k$ be such an integer.
The other is the property that all the packets available at a phase of \textsc{2-Full-Sensing$(k)$} are scheduled to be transmitted in this phase.
To see this, consider the following cases.
If there are $k$ available packets at a station at the start of a phase, then they are transmitted throughout the phase.
Otherwise either the number of available packets in both stations is $k$ or it is less than~$k$.
In the former case, the stations transmit back to back, starting with $p$, otherwise there is exactly one silent round when $p$ is already done but $q$ has not started transmitting yet.
We can afford the silent round since the number of available packets is less than the length $k$ of the phase.
It follows that once \textsc{2-Full-Sensing$(k)$} has been invoked, no collision ever occurs in the execution and so \textsc{2-Full-Sensing$(k+1)$} is never invoked.

There can be at most $w-1$ collisions because once \textsc{2-Full-Sensing$(w)$} is invoked there will be no collisions.
A collision results in unloading the old packets, which requires two additional silent rounds.
This contributes up to $3(w-1)$ void rounds.
Additional packet delay may occur due to the cumulative effect of the first phases consisting of void rounds of the procedures \textsc{2-Full-Sensing$(k)$} for consecutive values of~$k$.
That effect contributes a delay of at most~$w$.
This is because a transition to \textsc{2-Full-Sensing$(k)$} from \textsc{2-Full-Sensing$(k-1)$}, for $k>1$, begins by dequeuing up to $k$ entries form DELAY: when there are less than $k$ entries in DELAY in such a situation then up to $k$ void rounds are generated, but when there are at least $k$ of them then no extra void round occurs.
Thus the delay of a packet from injection to transmission is at most~$4w$.
\end{proof}

Next we show that an acknowledgment based algorithm for a system of two stations cannot be even stable.


\begin{theorem}
\label{thm:impossible-2ack-based-stable}
No acknowledgment based algorithm is stable in a system of two stations against the window adversary of burstiness~$2$.
\end{theorem}

\begin{proof}
Take an acknowledgment based algorithm~$\cP$.
We specify an execution~$\cE$ of~$\cP$ in which an infinite set of void milestone rounds occurs while the adversary of burstiness~$2$ injects a packet per round on the average.

Consider the first action performed when a new packet is started to be processed, determined by the first binary digit in a transmission sequence.
If there is a station that pauses at such a first round of processing a packet, which means that the transmission sequence begins with a~$0$, then the adversary may inject all the packets into this station only, which results in an infinite sequence of silent milestone rounds.
Otherwise each of the stations $p$ and $q$ transmits immediately, which means that the transmission sequences of each station start with an occurrence of~$1$.
From now on we consider this case only.
Observe that if a station transmits successfully and still has a packet, then the station will continue transmitting its packets until either a collision occurs or there is no pending packet, because these are all transmissions that have to occur at the first rounds of processing these packets, respectively.

Let the adversary inject a packet into each station at the first round of the execution, which results in both stations transmitting at the second round and so a collision. 
Define the second round to be the first milestone void round in~$\cE$.
Let the adversary pause at the second round.

Suppose we have defined a prefix of execution~$\cE$ that ends in a void milestone round~$t$, and such that each station has a pending packet at the end of round~$t$.
If there is a collision or silence at round~$t+1$, then let $t+1$ be the next milestone round, and let the adversary inject a packet at round~$t+1$ into an arbitrary station.
Otherwise only one station transmits at round~$t+1$, let it be~$p$.
Let the adversary keep inserting packets into both $p$ and~$q$ starting from round~$t+1$, according to the pattern $2,0,2,0,\ldots$, that is, two packets at round~$t+1$, then none at~$t+2$, then two at~$t+3$, then none at~$t+4$, and so on.
Station~$p$ will still have a packet available after the transmission at round~$t+1$, as a packet is injected at round~$t+1$ and therefore $p$ will immediately transmit again.
Transmissions of~$p$ will continue for as long as station~$p$ has packets available, for a minimum of two consecutive rounds.
Such transmissions of~$p$ cannot continue forever, as a packet gets injected into~$p$ at every other round while $p$ keeps transmitting continuously.
Therefore eventually either a collision occurs at some round~$t_1 > t+1$ or station~$p$ does not have a pending packet after a transmission of the only remaining packet at some round~$t_2>t+1$.

In the former case we can define round~$t_1$ to be the next void milestone round.
We need to suitably adjust the behavior of the adversary: if the adversary injects two packets at round~$t_1-1$, then let no packets be injected at~$t_1$, but if no packets are injected at round~$t_1-1$ then let the adversary inject one packet at round~$t_1$ into an arbitrary station.

Next consider the latter case in which $p$ transmits a packet at round $t_2>t+1$ and there are no more packets available at~$p$.
If $q$ does not transmit at round~$t_2+1$ then the silent round~$t_2+1$ can become the next milestone round, with a suitably adjusted behavior of the adversary.
What remains is a case when $q$ does transmit at round~$t_2+1$.
Observe that $t_2+1$ is a round when the adversary injects two packets, one per station, as otherwise $p$ would transmit at round~$t_2+1$ having a packet available that was injected at round~$t_2$.
Station~$q$ gets one packet injected at round~$t_2+1$, so $q$ has a pending packet at the end of round~$t_2+1$.
Consider round~$t_2+2$.
Station~$q$ transmitted at round~$t_2+1$ and it has a pending packet, so~$q$ transmits at round~$t_2+2$.
Simultaneously $p$ transmits at this round, having just obtained a new packet at the preceding round~$t_2+1$.
This results in a collision at round~$t_2+2$, and this round is declared to be the next milestone round of execution~$\cE$.
Let the adversary do not inject any packets at this round, having injected two packets at round~$t_2+1$.
The behavior of the adversary throughout the execution is consistent with the injection rate~$1$ for the window of size~$w=2$.
\end{proof}

\section{Three Stations}

\label{sec:three}

The algorithms presented in this section are developed for a system of \emph{exactly three} stations, while impossibility results hold for any system with \emph{at least three} stations.
We show that for window adversaries, there is an adaptive algorithm that handles injections of a packet per round with fair latency while no full sensing algorithm can provide even just stability.

The three stations are named $p$, $q$ and~$r$.
The stations are ordered $\langle p, q, r\rangle$ in a cyclic fashion; when we refer to the next station after a given one, we mean this cyclic ordering.
If $g$ is any station, then the station immediately following $g$ in this order is denoted by~$g'$,  and the station immediately following $g'$ is denoted by~$g''$.


\Paragraph{Adaptive  algorithm of fair latency.}

We start with an adaptive algorithm designed for a specific window adversary.
The algorithm is called \textsc{3-Adaptive-Window$(i)$}, where $i$ is interpreted as the window of the adversary.
An execution of the algorithm is structured as a sequence of consecutive phases.
A \emph{phase} consists of $i$ consecutive rounds, except for the first phase which takes $i+1$ rounds.
Packets injected during the $i$ rounds preceding the last round of the previous phase are called \emph{available} during the phase; such packets are determined for any phase after the first one.
A phase, starting from the second one, is used to transmit the available packets.
Stations may attach control information to be transmitted with packets, so the algorithm is adaptive.
In particular, when a station transmits the last available packet, then an ``over'' bit is attached to the message when needed to indicate this fact.
The ``over'' bit is sometimes needed to allow some other station to take over without a delay, as we will see next, so that the three stations can transmit back-to-back without intervening void rounds.

For each phase, there is a station designated to be the \emph{last} one for the phase.
The algorithm is structured such that the last station of the phase transmits at the last round of the phase.
If the last station still has a pending packet to transmit at the last round of the phase, then the packet is transmitted and a control bit is attached to the packet, otherwise only a control bit is sent.
This particular control bit is to indicate whether the last station has packets available for the next phase or not.
In the first phase, all the stations pause  through the first $i$ rounds.
Station~$p$ is designated to be the last one for the first phase so that it transmits at the last round~$i+1$ of the first phase.
Consider an arbitrary phase, called simply \emph{current}, and let $g$ be the last station of this phase.
We consider two cases, depending on whether $g$ has any packets available to be transmitted in the next phase.

If $g$ has packets available for the next phase, then $g$ starts the next phase with a sequence of transmissions of all its available packets.
If $g$ has exactly $i$ such packets, then $g$ is also the last station for the next phase.
Otherwise, when $g$ transmits its last available packet in the next phase, $g$ attaches the ``over'' control bit to the last packet.
After hearing the ``over'' signal, the stations~$g'$ and~$g''$ know how many rounds have remained in the next phase; let $k$ be this number.
If any station among $g'$ or $g''$ has $k$ packets available for the next phase, then this station unloads all its $k$ available packets in the remaining rounds of the next phase, and also gains the status to be the last station for the next phase.
Otherwise station~$g'$ starts unloading its available packets, if any, while $g''$ gains the status to be the last station for the next phase.
If $g'$ does not have any available packets, then $g'$ simply pauses, otherwise $g'$ attaches the ``over'' bit to the last transmitted packet.
When $g''$ hears either silence or the ``over'' signal, then $g''$ starts unloading its available packets, if any.
Finally, $g''$ transmits at the last round of the phase.

Next consider the case when $g$ has not received any packets to be available for the next  phase.
Station~$g$ has informed the remaining stations about this fact by the transmission at the last round of the current phase.
Now the situation is similar as in the previous case after $g$ sent the ``over'' signal, with $k=i$.
If any among $g'$ or $g''$ has received $i$ packets to be available for the next phase, then this station unloads all its $i$ available packets in the remaining rounds of the next phase, and also gains the status to be the last station for the next phase.
Otherwise station~$g'$ starts unloading its available packets, if any, while $g''$ gains the status to be the last station for the next phase.
Station~$g''$ takes over as soon as either silence or the ``over'' signal is heard.


\begin{lemma}
\label{lem:3-adaptive-window-w}

Algorithm \textsc{3-Adaptive-Window$(w)$} provides packet latency at most $2w+1$ against the adversary with window size $w$ in a system of three stations.
\end{lemma}

\begin{proof}
It is sufficient to show that all the packets available at the end of a phase are transmitted in  the next phase.
This is because there are $w$ rounds during which packets are injected to be available for the next phase, followed by the last round of the current phase, and $w$ rounds of the next phase.

The proof is by induction on the numbers of phases.
No packets are injected before the first phase, therefore the base of induction holds for the first phase by default.
Next we show the inductive step.
Consider some current phase and suppose that when the next phase starts, then there are no packets available for the current phase still in queues.
Consider a packet~$u$ available for the next phase.

If $u$ was injected into a station~$g$ that is the last one for the current phase, then station~$g$ indicates in the transmission at the last round of the current phase that it holds packets  available for the next phase.
This allows $g$ to unload all its packets, including $u$, starting from the first round of the next phase.

Next suppose that packet~$u$ resides at~$g'$.
If $g$ has no packets available for the next phase, then $g'$ knows this after the last round of the current phase  and so $g'$ may start unloading all its available packets, including~$u$, starting from the first round of the next phase.
Otherwise station~$g$ unloads its available packets first, and then station~$g'$ takes over without a delay after station~$g$ is done, and unloads all its available packets including~$u$.

Finally, suppose that $u$ was injected into~$g''$.
If $g''$ has $w$ packets available for the next phase, then $g$ has no such packets and $g''$ begins unloading its packets starting from the first round of the next phase, by the specification of the algorithm.
Suppose that $g''$ has $k$ packets available for the next phase, where $k<w$.
If $g$ has $w-k$ available packets, then after $w-k$ rounds of the next phase, station~$g''$ knows that it needs to transmit in the remaining $k$ rounds of this phase, which $g''$ does.
Otherwise, when $g$ has some $\ell<w-k$ available packets, possibly $\ell=0$, then eventually all these packets have been heard and $g'$ is to take over.
Station~$g'$ has at most $w-k-\ell$ packets available for the next phase, so after the packets of $g$ and $g'$ have been heard and station~$g''$ takes over, there are at least $k$ rounds in the phase, so packet~$u$ is transmitted in one of them.
These cases exhaust all the possibilities, which shows that $u$ is transmitted eventually in the next phase.
This completes the proof of the inductive step.
\end{proof}

We show next that there is a stable and fair algorithm for three stations that does not rely on the knowledge of the window but but resorts to the mechanism of collision detection.
The algorithm is called \textsc{3-Adaptive-Col-Det}.
It is obtained by modifying algorithm \textsc{3-Adaptive-Window($i$)} as follow.
Initially algorithm \textsc{3-Adaptive-Col-Det} runs \textsc{3-Adaptive-Window($1$)} to try the window $w=1$.
When a collision occurs while running \textsc{3-Adaptive-Window($i$)}, then \textsc{3-Adaptive-Window($i+1$)} is invoked.
We apply a similar approach as in Section~\ref{sec:two} by using a queue called DELAY.
Packets stored already in the queues at a round  when \textsc{3-Adaptive-Window($i+1$)} is invoked are called \emph{old}.
First the old packets are transmitted, by the stations $p$, $q$ and $r$, in this order, while records of actions of the adversary are enqueued, with an entry for each round.
A transition to the next station is indicated either by a control bit ``over'' attached to the last old packet of a station or by only this bit transmitted when the station does not have any old packets.
After all the old packets have been heard, the stations start executing \textsc{3-Adaptive-Window($i+1$)}.
To this end, they keep dequeuing DELAY, if it is nonempty, while simultaneously enqueueing the current actions of the adversary.
Details are as follows.
An invocation of \textsc{3-Adaptive-Window$(i+1)$} starts with the first phase of $i+1$ rounds. 
During the $i$ initial rounds the stations pause storing up to $i$~packets to be available for the next phase.
When multiple invocations of \textsc{3-Adaptive-Window($i$)} are performed back-to-back, for consecutive values of $i$, we can modify the action performed during a transition from $i$ to $i+1$ to save on void rounds: let the stations begin  \textsc{3-Adaptive-Window($i+1$)} by dequeuing either $i+1$ entries from the queue DELAY or all of them, whichever is smaller, and let this contribute towards the beginning of the first phase of \textsc{3-Adaptive-Window$(i+1)$}.
Next the stations dequeue one entry per round, unless DELAY is empty, in which case it is not used.


\begin{lemma}
\label{lem:3-adaptive-col-det}

Algorithm \textsc{3-Adaptive-Col-Det} is of fair latency in a system of three stations with a channel with collision detection.
\end{lemma}

\begin{proof}
The proof is similar to that of Theorem~\ref{thm:2-full-sensing-bounded-latency}, with some arguments provided by Lemma~\ref{lem:3-adaptive-window-w}.
There are at most $w-1$ collisions before the correct window size is reached.
Each collision triggers unloading old packets.
Each such an instance of transmitting old packets may cause at most one void round because at most one station holds no packets.
The cumulative effect of first phases consisting of void rounds of procedures \textsc{3-Adaptive-Window($i$)}, for consecutive values of~$i$, contributes an additional delay of at most~$w$.
This is because transition to $i$ from $i-1$, for $i>1$, begins by dequeuing up to $i$ entries form DELAY: when there are less than $i$ entries in DELAY, then up to $i$ void rounds are generated, but when there are at least $i$ of them then no extra void round occurs.
By Lemma~\ref{lem:3-adaptive-window-w}, after the size of the window has been set correctly, packet delay is at most $2w+1$.
Thus the total delay of a packet is at most $w-1+w-1+w +2w+1 < 5w$.
\end{proof}

Next we present the ultimate algorithm for three stations for the channel without collision detection, it is called \textsc{3-Adaptive}.
We simulate algorithm \textsc{3-Adaptive-Col-Det} by detecting collisions by silences.
Consider a phase that is at least second after an invocation of \textsc{3-Adaptive-Window($i$)} in \textsc{3-Adaptive-Col-Det}, for $i\ge 1$.
If a collision occurs at a round of such a phase that is not the last one in the phase, then there are only two stations involved, say, $g$ and~$g'$, by the design of \textsc{3-Adaptive-Col-Det}.
In such a case, the station~$g$ stops transmissions in this phase while station~$g'$ continues by repeating the last transmission.
At this point the station~$g''$ may not know about the collision.
To allow $g''$ learn this, we use the property of \textsc{3-Adaptive-Col-Det} that some station transmits at the last round of every phase.
Let every station that knows about a collision in a phase transmit at the last round of the phase, the contents being a dummy message for a station that is not designated to be the last one for the phase.
All the stations hear silence at this round, so all of them learn of collision, which triggers an invocation of \textsc{3-Adaptive-Window($i+1$)}.
If the first collision in a phase occurs at the last round of the phase, then this automatically makes every station learn about the collision by the silence heard, which immediately triggers an invocation of \textsc{3-Adaptive-Window($i+1$)}.


\begin{theorem}
\label{thm:3-adaptive-fair-latency}

Algorithm \textsc{3-Adaptive} is of fair latency for a system of three stations against window adversaries.
\end{theorem}

\begin{proof}
A proof similar to that of Lemma~\ref{lem:3-adaptive-col-det} applies.
The difference is in the number of void rounds in a phase of \textsc{3-Adaptive-Window($i$)}, for a suitable~$i$.
This is because algorithm \textsc{3-Adaptive} has stations complete each phase during which a collision occurs.
This contributes one extra void round in a phase, namely the one occurring at the last round of the phase.
Thus the delay of a packet is at most $w$ more than the bound of Lemma~\ref{lem:3-adaptive-col-det}.
\end{proof}


\Paragraph{Impossibilities for at least three stations.}

Next we show that there is no stable full sensing algorithm for a system of three stations.


\begin{theorem}
\label{thm:impossible-3-stations-full-sensing-stable}

No full sensing algorithm can be stable for three stations against the window adversary of burstiness~$2$.
\end{theorem}

\begin{proof}
Consider a full sensing algorithm $\cP$ for a system of three stations $p$, $q$ and~$r$.
We argue similarly as in the proof of Theorem~\ref{thm:impossible-2-full-sensing-stable} by  identifying milestone rounds $t_{0}, t_{1},t_{2},\ldots$ in an execution of $\cP$, where $t_{i}<t_{i+1}$ for any $i\ge 0$, while maintaining the average injection rate~$1$.

Define $t_{0}$ to be the first round of execution~$\cE$ of~$\cP$.
Let the adversary inject one packet into~$p$ at the first round. 
Round~$t_0$ is void in~$\cE$ because no packets have been injected before the round.
Suppose the execution~$\cE$ and injections have been defined by a milestone round~$t_{i}$, for some $i\ge 0$: we determine the milestone round~$t_{i+1}>t_{i}$ and specify how packets are injected in the interval starting from~$t_{i}+1$ through~$t_{i+1}$ in~$\cE$.

We consider a sequence of conceptual scenarios beyond round~$t_i$.
The adversary injects packets according to the pattern $2,0,2,0,2,\ldots$ in each of the scenarios.
This means that two packets are injected at round~$t_i+1$, then no packets at all at round~$t_i+2$, then two packets at round~$t_i+3$, and so on.
The adversary always uses the same two stations to inject a packet per station according to this pattern.
We consider three scenarios corresponding to three pairs of stations possible to select out of $p$, $q$ and~$r$.
The adversary injects packets into stations~$p$ and~$q$ only in scenario~$\cS_1$, into stations $p$ and~$r$ only in scenario~$\cS_2$, and into stations $q$ and~$r$ only in scenario~$\cS_3$.
Let $t'\ge t_i$ be the first round such that station~$r$ does not transmit after round~$t'$ under scenario~$\cS_1$, and $q$ does not transmit after~$t'$ under~$\cS_2$, and $p$ does not transmit after $t'$ under~$\cS_3$.
Such a $t'$ exists as we refer to a station that the adversary does not inject into after~$t_i$ under the respective scenario.

We show that a void round is bound to occur in one of the scenarios.
Suppose that it is not the case to arrive at a contradiction.
Consider the first scenario~$\cS_1$ in which the adversary injects packets only into stations~$p$ and~$q$.
Station~$q$ transmits at some round~$t''>t'$, due to the stability of the algorithm, while station~$p$ pauses at round~$t''$.
Consider the second scenario~$\cS_2$ in which the adversary injects packets only into the stations $p$ and~$r$.
Station~$p$ does not transmit at round~$t''$ under scenario~$\cS_2$, since $p$ paused at round~$t''$ in scenario~$\cS_1$ and scenario~$\cS_2$ is perceived by~$p$ as identical with~$\cS_1$ up to round~$t''$ in both the feedback from the channel and the pattern of packet injections, as there are no void rounds.
Therefore station~$r$ transmits at round~$t''$ under scenario~$\cS_2$.
Finally consider scenario~$\cS_3$ in which the adversary injects only into stations~$q$ and~$r$.
Station~$q$ behaves in the same way up to round~$t''$ under this scenario as under scenario~$\cS_1$ as it cannot detect any difference in both the feedback from the channel and the pattern of packet injections, as there are no void rounds.
Similarly, station~$r$ behaves in the same way up to round~$t''$ under this scenario as under scenario~$\cS_2$.
It follows that both stations $q$ and~$r$ transmit at round~$t''$ which results in a collision.
This creates a void round.

Take one scenario~$\cS$, from among the scenarios $\cS_i$, for $1\le i\le 3$, that produces a void round~$t>t_i$ the soonest.
Define $t_{i+1}=t$ and extend the execution $\cE$ starting from the round~$t_i+1$ through $t=t_{i+1}$ as follows.
Let the adversary inject packets according to scenario~$\cS$ up to round~$t - 1$.
If two packets are injected at round~$t-1$ then let the adversary pause and not inject any packet at round~$t$, otherwise when no packets are injected at round~$t-1$ then let the adversary inject one packet into any station at round~$t_{i+1}=t$.

It follows from the construction that all the milestone rounds~$t_i$ are void while the average number of packets injected between $t_i$ and~$t_{i+1}$ is one packet per round.
We obtain that the number of packets at the queues of the three stations at round~$t_{i}$ is at least~$i+1$, for $i \ge 0$.
This yields the existence of an execution~$\cE$ of $\cP$ in which the number of packets in queues grows unbounded, so the algorithm is unstable.
The way packets are injected is consistent with the definition of an adversary of burstiness~$2$.
\end{proof}

We show next that no algorithm that withholds the channel can be stable for a system with at least three stations.


\begin{theorem}
\label{thm:impossible-3-stations-withhold-stable}

No algorithm that withholds the channel can be stable for three stations against the leaky-bucket  adversary of burstiness~$2$.
\end{theorem}

\begin{proof}
Consider an algorithm $\cP$ that withholds the channel for a system of three stations $p$, $q$ and~$r$.
We argue by  identifying void milestone rounds $t_{0}, t_{1},t_{2},\ldots$ in an execution of $\cP$, where $t_{i}<t_{i+1}$ for any $i\ge 0$, while maintaining the average injection rate~$1$. 

Define $t_{0}$ to be the first round of execution~$\cE$ of~$\cP$.
Let the adversary inject one packet into any station at the first round; round~$t_0$ is void in~$\cE$.
Suppose the execution~$\cE$ and injections have been defined by a milestone round~$t_{i}$, for some $i\ge 0$: we define the milestone round~$t_{i+1}>t_{i}$ and how packets are injected in the interval of~$\cE$ starting from~$t_{i}+1$ through~$t_{i+1}$.
If a void round occurs in the course of the described construction, then this gives a milestone round and we stop to determine $t_{i+1}$ and injections up to this round.
In what follows we consider cases in which the current round is not a void one, whenever logically possible.

We first make the queues at two stations empty in the following way.
Let the adversary keep injecting packets into~$p$ only, a packet per round, starting from round~$t_i+1$.
Eventually $p$ broadcasts, because algorithm~$\cP$ is stable.
At a round when $p$ transmits for the first time, switch injecting to~$q$ and keep injecting packets into~$q$ only, a packet per round.
Eventually $q$ broadcasts, because algorithm~$\cP$ is stable.
At a round when $q$ transmits for the first time, switch injecting to $r$ and keep injecting packets into~$r$ only, a packet per round.
Eventually $r$ broadcasts, because algorithm~$\cP$ is stable.
At the round when $r$ takes over, let the adversary inject only into either~$p$ or~$q$, a packet per round.
Consider the round~$t$ when $r$ transmits its last packet and its queue becomes empty.
At this point there are a number of packets in the queues of~$p$ and~$q$.
We consider a sequence of conceptual scenarios beyond that round, depending on the partitioning of the packets between $p$ and~$q$.

The first scenario~$\cS_0$ is such that all the recently injected packets are in the queue of~$p$ while the queue of~$q$ is empty.
This means that $p$ needs to take over from $r$ after hearing its last packet at round~$t$, since otherwise there would be a void round.

Consider another scenario~$\cS_1$, with the only difference with respect to $\cS_0$ being in that $q$ has one packet.
This is possible because of burstiness~$2$ of the adversary.
These two scenarios are identical for~$p$ at the round~$t$ when $r$ transmits the last packet, so in such a scenario~$p$ also transmits at round~$t+1$ just after~$r$.

Consider a scenario~$\cS_2$ which differs from~$\cS_1$ only in that the number of packets at~$p$ is one less than in~$\cS_1$ while $q$ has also just one packet.
The adversary does not need to use burstiness to obtain such a scenario as this requires injecting precisely one packet per round while $r$ is transmitting.
Since $q$ cannot see any difference between $\cS_2$ and~$\cS_1$ at the round~$t$ of the last transmission of~$r$, again $q$ pauses while $p$ broadcasts at round~$t+1$ just after station~$r$, or otherwise there would be a void round.

Consider a scenario~$\cS_3$ which differs from~$\cS_2$ only in this that $q$ has two packets rather than only one at the last round~$t$ when $r$ broadcasts.
This is possible to achieve due to the burstiness~$2$ of the adversary.

These two scenarios $\cS_2$ and~$\cS_3$ are identical for~$p$ at the round~$t$, so in such a scenario~$p$ also transmits at round~$t+1$, or otherwise the round would be void.

We continue through a sequence of scenarios to obtain a scenario~$\cS_4$ such that  station~$p$ has just one packet while $q$ has all the remaining packets.
Again station~$p$ transmits for the first time at round~$t+1$ immediately after $r$ while $q$ pauses at this round.

The next scenario~$\cS_5$ has the adversary not inject the one packet into~$p$ but only the packets injected into~$q$ in scenario~$\cS_4$.
Station~$q$ cannot notice a difference between $\cS_5$ and~$\cS_4$ at the last round~$t$ of $r$ transmitting, so it pauses at round~$t+1$ as in~$\cS_4$.
This results in both $p$ and~$q$ pausing at round~$t+1$, which results in round~$t+1$ being silent.
We have thus shown that there is a scenario for the adversary to enforce a void round.

Take one scenario~$\cS$ of injections of the adversary, from among the scenarios discussed above, that produces a void round~$t'$ the soonest.
The remaining part of the argument is similar to the conclusion of the proof of Theorem~\ref{thm:impossible-3-stations-full-sensing-stable}.
\end{proof}

\section{Many Stations}

\label{sec:four}

We develop an adaptive algorithm that is stable against leaky-bucket adversaries for any number~$n$ of stations.
The queues at stations can grow up to $\Omega(n^2+\text{burstiness})$, which we show to be unavoidable.
The algorithm uses the queue sizes at stations in an essential way, in that they affect state transitions.
Next we show that stability cannot be achieved for certain restricted algorithms, against window adversaries in systems with at least four stations, in particular for algorithms that are queue-size oblivious.


\Paragraph{Stable algorithm.}

We call \textsc{Move-Big-To-Front}$(n)$ an adaptive algorithm to be presented next.
The number of stations~$n$ affects the course of action each time a station transmits, we emphasize this by including $n$ as a parameter in the name of the algorithm.
The algorithm schedules exactly one station to transmit at each round, so collisions never occur.
This is implemented by using a conceptual ``token'' giving the right to transmit, which is assigned in such a way that at each round exactly one station holds the token.

Every station maintains a list of all the stations in its private memory.
The list is initialized as sorted in the increasing order by the names of the stations.
The operations performed on the lists are determined uniquely by what has been heard on the channel.
Hence all these lists at stations are manipulated in exactly the same way.
This guarantees that the lists are identical in all stations at all rounds.
Because of this property, we refer to all these lists as copies of \emph{the} list.
Initially the first station in the list holds the token.

The algorithm is executed at a given round as follows.
A station~$p$ with the token broadcasts a packet, if it has any.
If the station with the token does not have a pending packet, then the station does not transmit, which results in a silent round.
A station considers itself \emph{big} at a round when it has at least $n$ packets available.
A big station attaches a control bit to indicate this status to each packet it transmits while big.
After a station announces itself to be big,  it is moved to the front of the list and keeps the token for the next round.
After a station with a token broadcasts while it is not big or when it does not transmit at all, then the token is moved immediately to the next station in the list.
Here being ``next'' is understood in the cyclic ordering of the list of stations, in that the token from the last station in the list is moved directly to the first one.

Algorithm \textsc{Move-Big-To-Front} resembles \textsc{Round-Robin} in that when queues are small then the token traverses the list in a cyclic fashion: when a station~$p$ is followed immediately by $q$ in the list and $p$ holds less than $n$ packets at a round of transmission, then $q$ obtains the token immediately after the transmission by~$p$.
A difference with \textsc{Round-Robin} is in the possibility to have the token hop to the front of the list: when a station~$p$ is big at the time of its transmission, then $p$ still holds the token after being moved to the front of the list, so that when $p$ eventually releases the token, it is the second station in the list, the one directly following the station~$p$ at the front of the list, that transmits after~$p$.

Consider a scenario in which, starting from some round, the adversary injects packets into one station only.
Such a station is eventually detected to be big and then this station keeps transmitting throughout the remaining part of the execution, while the other stations are neglected.
It is a scenario in which some packets are never transmitted, so the algorithm is not fair.

The design of this algorithm is based on the following intuitions how to provide stability.
Define a \emph{pass of the token} to be a traversal of the token starting at the front of the list and ending either at a new big station or again at the front station of the list after traversing the whole list, whichever occurs first.
Define a \emph{life cycle} of a station to be a time period which starts either at the first round of the execution or at a round when the station is discovered to be big, and which ends just before the station is discovered to be big again, when this happens.
Suppose there are at least $n^{2}$ packets in queues of stations at the beginning of a pass of the token.
Some station has at least $n$ packets, by the pigeonhole principle.
It follows that a new big station is discovered during this pass of the token.
Because it takes at least $n-1$ token passes for a former big station to have its queue empty, a former big station either maintains a nonempty queue during its life cycle, or it drifts towards the end of the list so that eventually it will not be visited by the token at all, unless the adversary is lazy and does not inject as many packets as possible.
It follows that if the adversary keeps injecting  at full power of one packet per round, then eventually a round occurs such that afterwards packets are heard at all rounds.
We make these ideas precise in the proof of the next theorem.


\begin{theorem}
\label{thm:move-big-to-front-stable}

If algorithm \textsc{Move-Big-To-Front}$(n)$ is executed against the leaky-bucket adversary of burstiness $b+1$, then the number of packets stored in queues is at most $2(n^2 + b)$ at any round.
\end{theorem}

\begin{proof}
Suppose, to the contrary, that there is a round with at least $2n^2+2b+1$ packets in queues.
There is a time segment~$T$ with the following properties:
\begin{enumerate}
\item[(i)] there are at least $n^2$ and at most $n^2+b$ packets in queues at the beginning of $T$,
\item[(ii)]  there are at least $n^2$ packets in queues at each round of $T$, and
\item[(iii)] there are at least $2n^2+2b+1$ packets in queues at the end of $T$.
\end{enumerate}
In the remaining part of the proof we restrict our attention only to the rounds in~$T$.
The notions of a pass of the token and of a life cycle of a station are relativized to~$T$.

Consider a pass of the token.
A new big station is eventually found during this pass, because at least one station has at least  $n$ packets in its queue, by the pigeonhole principle combined with property~(ii) of~$T$.

Let $C$ denote the set of all the stations that are discovered at least once to be big during a round in~$T$.
If a station~$p$ is not in $C$, then $p$ will eventually drift through the list to be located behind all the stations in~$C$.
When a token passes through $p$ and there are no packets in~$p$, then this results in a silent round.
We assume the worst case when each event of receiving the token by~$p$ results in a silent round.
Let $q$ be the station discovered to be big in this pass of the token.
Station~$q$ is moved to the front of the list and $q$ will never again be behind~$p$, so that $q$ can be associated with exactly one silent round of each such a station~$p$ not in~$C$.
Since there are $|C|$ such stations $q$ and $n-|C|$ such stations~$p$, the total contribution of the stations that are not in~$C$ to the number of silent rounds is at most $(n-|C|)\cdot |C|$.

We claim that once a station~$q$ is discovered to be big, then $q$ transmits a packet each time $q$ holds a token.
To show that this is the case, consider a life cycle of~$q$.
During a pass of the token, either $q$ is discovered to be big again, which starts a new life cycle for~$q$ with at least $n-1$ packets still remaining in the queue, or one station in~$C$ located behind~$q$ is discovered as big.
The latter event results in the number of stations in $C$ behind $q$ in the list decreasing by one.
Since there are at most $|C|-1 < n$ stations in $C$  behind $q$ in the list, station~$q$ is visited at most $|C|-1<n$ times by the token during the life cycle of~$q$.
It follows that after each station in~$C$ has been discovered at least once to be big, no station in $C$ ever has an empty queue.

It remains to estimate the number of silent rounds before a station~$q$ in~$C$ becomes big for the first time in~$T$.
Notice that $q$ could obtain the token with an empty queue at most $|C|-1$ times before this happens.
This is because each time $q$ has the token with an empty queue, there is some station~$p$ from $C$ behind~$q$ on the list such that $p$ is discovered as big in this pass of the token.
The discovery results in moving~$p$ to the front of the list, so that $p$ stays before $q$ until $q$ becomes big for the first time.
There are $|C|$ such stations $q$ and $|C|-1$ such stations $p$.
Therefore the number of silent rounds contributed by all the stations in $C$ is at most $|C|\cdot (|C|-1)$.

To sum up, the total number of silent rounds in $T$ is at most
\begin{equation}
\label{eqn:silent}
(n-|C|)\cdot |C|+(|C|-1)\cdot |C| < n\cdot |C| \le n^2\ .
\end{equation}
The difference between the number of injected packets and the number of transmitted packets equals the number of void rounds plus burstiness, which is at most $n^2+b$ by~\eqref{eqn:silent}.
Combine this fact with property (i) of $T$ to obtain $(n^2+b)+(n^2+b)=2n^2+2b$ as an upper bound on the number of packets in the system at the end of $T$.
This contradicts property (iii) defining~$T$.
\end{proof}

Next we show, for any broadcast algorithm for $n$ stations, that the system may be forced into a configuration with $\Omega(n^2)$ packets in queues.


\begin{theorem}
\label{thm:lower-bound-n-squared}

For any algorithm for $n$ stations, the leaky-bucket adversary of burstiness $2$ can enforce an execution such that eventually there are at least $\left(\frac{n}{2}-1\right)^2$ packets in the queues at the stations.
\end{theorem}

\begin{proof}
Let us take an arbitrary algorithm~$\cP$ for $n$~stations and fix it for the remaining part of the proof.
We determine an execution~$\cE$ of algorithm~$\cP$ by identifying milestone rounds.
Similarly as in such previous constructions, the underlying principle is to maintain a property  that each milestone round is void while the average injection rate of the adversary is~$1$.
The difference with previous impossibility proofs is that the sequence of milestone rounds may be either finite or infinite.
We  denote by $t_i$ the $i$th milestone round, when it exists for this~$i$.
Milestone rounds will have the property that the number of packets in queues at the beginning of round~$t_i$ that exists is at least $i$, for $i\ge 0$.
It follows that if algorithm~$\cP$ is stable, then the sequence of milestone rounds will be finite.

If the algorithm is unstable, then it is sufficient to take an execution in which the queues grow unbounded.
Therefore we assume that the algorithm is stable in the remaining part of the proof.
We will construct an execution~$\cE$ of this algorithm by specifying a behavior of the adversary and a resulting sequence of milestone rounds.
Set the first round of the execution~$\cE$ to be the milestone round~$t_0=1$.
Let the adversary inject one packet into some station at round~$t_0$.
Suppose that a milestone round~$t_i$ has been determined, for some $i\ge 0$, together with injections up to this round.
We assume the invariant that the number of packets injected by round~$t_i$ is at most one packet per round when averaged over all the intervals ending at~$t_i$. 
We show how packets are to be injected so that either the next milestone round~$t_{i+1}$  becomes defined or the whole execution becomes completely defined at a stretch.

The injections of the adversary are to be considered as conceptual only in that we search for possible extensions of execution~$\cE$ after~$t_i$.
As soon as a next milestone round~$t_{i+1}$ has been determined, we extend the execution~$\cE$ until this round~$t_{i+1}$ by having the adversary actually inject packets at rounds from~$t_i+1$ through~$t_{i+1}-1$ according to the pattern that determined round~$t_{i+1}$, and the number of packets injected at the new milestone round~$t_{i+1}$ is to be maximum to provide the invariant.
The simplest case is when round~$t=t_i+1$ is void: then let the adversary inject one packet into any station at round~$t$ which becomes the next milestone round~$t_{i+1}$.
When the above is not the case, meaning a packet is heard on the channel at round~$t=t_i+1$, then we systematically explore other scenarios.

Let us begin with the segment $[t,t+1]$ of two consecutive rounds.
If there is a station~$p$ such that injecting into~$p$ at round~$t$ would result in round~$t+1$ being void, then let the adversary perform such an injection so that $t+1$ becomes the next milestone round~$t_{i+1}$.
Otherwise, if there are two stations $p$ and~$q$ such that injecting a packet into each of them at round~$t$ results in both of them transmitting at round~$t+1$, then let the adversary perform these two injections so that $t+1$ becomes the next milestone round~$t_{i+1}$.
Otherwise, if none of the above cases holds, then we have a situation that a unique station would transmit at round~$t+1$ if the adversary injected a packet into any station at round~$t$.
Suppose this is the case and let $v_1$ be the unique station that would transmit at round~$t+1$ as determined above.
Identifying a station~$v_1$ is conceptual only at this point.

Next consider the segment $[t,t+1,t+2]$ of three consecutive rounds.
For any stations $p$, $q$ and~$r$, different from~$v_1$, including the cases of repetitions among these three names of stations, consider injecting a packet into~$p$ at round~$t$ and a packet per station into both $q$ and~$r$ at round~$t+1$.
If this results in either a void round~$t+1$ or $t+2$, for some such stations $p$, $q$ and $r$, then let the adversary perform these injections up to the first void round so that this round  becomes the next milestone one~$t_{i+1}$.
Otherwise the station into which a packet is injected at round~$t$ uniquely determines the station that transmits at round~$t+2$, as long as we also inject a single packet at round~$t+1$.
Suppose that a packet per round is injected at rounds $t$ and $t+1$ into the smallest station~$p$ different from $v_1$: let $v_2$ be the only station that would transmit at round~$t+2$.
It may be the case that stations $v_1$ and~$v_2$ are actually the same station.
Identifying stations $v_1$ and~$v_2$ is conceptual only at this point.

We consider yet another initial segment $[t,t+1,t+2,t+3]$ of four consecutive rounds, in order to clearly see a pattern emerging in such constructions.
For any stations $p$, $q$, $r$ and~$s$, different from $v_1$ and~$v_2$, including  repetitions among these four names $p$, $q$, $r$ and~$s$ of stations, consider injecting a packet into~$p$ at round~$t$, next a packet into~$q$ at round~$t+1$, and finally a packet per station into both $r$ and $s$ at round~$t+2$.
If this results in a void round up to round~$t+3$, for some such stations $p$, $q$, $r$ and $s$, then let the adversary perform these injections so that the first such a void round becomes the next milestone round~$t_{i+1}$.
Otherwise the stations into which packets are injected at rounds $t$ and~$t+1$ uniquely determine the station that transmits at round~$t+3$, as long as we also inject a single packet at round~$t+2$.
Suppose a packet per round is injected into the smallest station different from $v_1$ and $v_2$ at rounds $t$, $t+1$, and $t+2$, and let $v_3$ be the only station that transmits at round~$t+3$.
The sequence $\langle v_1,v_2,v_3\rangle$ of stations may include repetitions.
Identifying stations $v_1$, $v_2$ and~$v_3$ is conceptual only at this point.

We continue in this way, which results in determining either a milestone round and extending the execution~$\cE$ through this round or a sequence~$\langle v_i\rangle $ of stations for a  segment of indices $i\ge 1$.
Next we consider only the case of sequence~$\langle v_i\rangle $ of stations.
Such a sequence has the property that the stations occurring as entries of the sequence transmit at consecutive rounds in the order in which they occur in the sequence, that is, station~$v_1$ transmits at round~$t+1$, station~$v_2$ transmits at round~$t+2$, station~$v_3$ transmits at round~$t+3$, and so on.
Moreover, this sequence of stations is defined in such a way that as soon as the stations~$v_i$ become determined for all $1\le i\le j$, then we conceptually inject into a station different from any among these $v_i$ for $1\le i\le j$ in order to determine the next station~$v_{j+1}$.

Can the sequence~$\langle v_i\rangle $ of stations be extended indefinitely with a void round never occurring?
This is impossible for a stable algorithm, by the following argument.
Suppose, to arrive at a contradiction, that $\cU$ and $\cW$ are two disjoint and nonempty sets of stations such that any station~$v_i$ in the sequence belongs to~$\cU$, while the set~$\cW$ contains the remaining stations.
Let $j$ be the largest index of a station~$v_j$ in~$\cU$ such that $v_i\ne v_j$ for $i<j$.
The stations $v_{k}$ in the sequence~$\langle v_i\rangle $ for $k>j$ are determined in such a way that packets are conceptually injected only into a station~$v\in \cW$ of the smallest name among the stations in~$\cW$.
If this could be continued forever, then we could keep injecting only into a station~$v$ that would never transmit after round~$t-1$, as $w$ is not in $\cU$, so the queue would grow unbounded at~$v$.

Therefore eventually either a milestone round manifests itself or a sequence $\langle v_i\rangle_{1\le i\le j}$ cannot be extended, for some positive integer~$j$.
The latter occurs when \emph{every} station already occurs at least once as some $v_i$, for $1\le i\le j$.
As soon as such a sequence $\langle v_i\rangle_{1\le i\le j}$ is determined that includes the names of all stations, then we say that a \emph{stage} has been completed.
Observe that the sequence~$\langle v_i\rangle$, for $1\le i\le j$ is defined in such a way that if the adversary injected $j$ times into~$v_j$ a packet per round starting from round~$t$, then the stations~$v_i$ would transmit in the order of their indices in the sequence, for $1\le i\le j$.
The station~$v_j$ is said to be the \emph{pivot station} of the stage.

After completing the first stage, we proceed in the same manner to complete next stages, starting a new stage with an empty sequence~$\langle v_i\rangle$.
Let us stop after we have completed $\lfloor \frac{n}{2}\rfloor$ stages, unless a void round occurs earlier.
The following properties of a stage are relevant: every station transmits at least once during a stage, while the adversary injects packets only into the pivot station of the stage.
There are at least $\frac{n}{2}-1$ stations that are not pivot for any stage.
Each one among these stations transmits at least once during any stage, which means transmitting for a total of at least $\frac{n}{2}-1$ times during all these stages.
Since the adversary does not inject any packets into non-pivot stations during these stages, the packets the non-pivot stations would transmit must already reside in their queues when the first stage starts.
It follows that the number of packets in queues of the stations  at the round~$t$ is at least $\left(\frac{n}{2}-1\right)^2$.
Once there are these many packets in queues at round~$t$, milestone rounds are not needed at all after round~$t$, as a trailing part of the execution can be determined at a stretch.
Namely, the adversary may inject a packet per round at arbitrary stations starting from round~$t$ and the number of packets in queues will remain at least $\left(\frac{n}{2}-1\right)^2$ at all rounds after~$t$.
\end{proof}


\Paragraph{Retaining algorithms.}

We know that no algorithm can be both stable and fair against leaky-bucket adversaries in systems of at least two stations, but fair latency is achievable in systems of up to three stations against window adversaries.
A question if achieving fair latency against window adversaries in systems of more than three stations is not settled by these facts.
We answer this question in the negative next.
Our approach is to show that achieving stability  against window adversaries in systems of more than three stations is impossible for a class of algorithms that have a property that  generalizes fairness and withholding the channel.

An algorithm is called \emph{retaining} if at any round when a station, say, $p$ transmits a packet successfully and when starting from this round the adversary injects packets only at other stations, then eventually station~$p$ does not have a pending packet.
Observe that an algorithm that is either fair or that withholds the channel is retaining.


\begin{theorem}
\label{thm:impossible-retaining-stable}

No retaining algorithm is stable in a system of at least four stations against the window adversary of burstiness~$2$.
\end{theorem}

\begin{proof}
Consider a retaining broadcast algorithm~$\cP$.
Choose some four stations $p$, $q$, $r$, and~$s$.
The adversary will inject packets only into these four stations.
The other stations cannot transmit packets so their transmissions may be ignored.
We define an execution in which the adversary injects one packet per round on the average
while a certain sequence $t_0, t_1, t_2,\ldots$ of void milestone rounds is unbounded.

Let the adversary inject a packet into some station at the first round. 
Define $t_0$ to be the first round, which is void.
Suppose we have defined the execution and the injections up to a void milestone round~$t_i$.
We need to specify what happens starting from the next round~$t=t_i+1$.
If round~$t$ is void then define $t_{i+1}=t$ and let the adversary inject a packet at some station at this round.
Otherwise exactly one station, say $p$, is to transmit at round~$t$.
We consider the consecutive rounds starting from~$t$ one by one, determining the injections.
There are the following two cases that help to structure the argument.
The meaning of $t$ will sometimes be of a round number larger than $t_i + 1$ when we need to repeat the construction of a segment of an execution.

\noindent
\textsf{Case 1:} At least two stations have pending packets at the beginning of round~$t$.

The adversary chooses a station~$q$ different from $p$ but also with a pending packet.
The adversary keeps injecting a packet per round into~$q$ starting from round~$t$.
This continues until a round~$t'>t$ occurs that is either void or at which station~$p$ does not have a pending packet.
One of these cases has to occur because the algorithm is retaining.
In the former case we define the round~$t'$ to be milestone round~$t_{i+1}$ and the adversary injects a packet into any station at that round.
In the latter case we may need to go through the same case, but with fewer nonempty queues; it is here that $t$ may acquire the meaning of a certain round number larger than~$t_i+1$. 
Eventually we will proceed as in the next case.

\noindent
\textsf{Case 2:} Only station~$p$ has a pending packet at the beginning of round~$t$.

We consider a number of conceptual scenarios to identify a void round~$t'>t$.
In the first one the adversary keeps injecting one packet into station~$q$ and another into station~$r$ at every other round.
This means the pattern $2,0,2,0,\ldots$, in terms of the numbers of packets injected, starting from round~$t$.
Eventually station~$p$ does not have a pending packet, as the algorithm is retaining and $p$ transmitted at round~$t$.
If a void round~$t'>t$ does not occur under this scenario, then either station~$q$ or station~$r$ makes the first successful transmission at some round~$t'>t$.
This round~$t'>t$ is the first one after~$t$ at which $p$ does not transmit.
Suppose it is station~$r$ that transmits while station~$q$ pauses.
None among the stations~$p$ and~$s$ transmits at round~$t'$ as the transmission of~$r$ is heard. 

Consider the second conceptual scenario in which station~$r$ is replaced by~$s$, that is, we consider the pair of $q$ and~$s$ rather than $q$ and~$r$.
Eventually station~$p$ does not have a pending packet, as the algorithm is retaining.
If a void round~$t'>t$ does not occur under this scenario then either station~$q$ or station~$s$ makes the first successful transmission at some round~$t'>t$.
This round~$t'>t$ is the first one after~$t$ at which $p$ does not transmit so neither station~$p$ nor~$q$ can distinguish between the two scenarios up to round~$t'$.
Therefore both $p$ and $q$ pause at round~$t'$ in both scenarios while station~$s$ transmits successfully under the second scenario.

The third scenario has the stations~$r$ and~$s$ play the role of the stations~$q$ and~$r$ in the first scenario.
Eventually station~$p$ does not have a pending packet, as the algorithm is retaining.
Let $t'$ be the first round after~$t$ at which $p$ does not transmit.
This is the same round~$t'$ with this property as under the first and second scenarios, since $p$ cannot distinguish between these three scenarios up to round~$t'$, being the only transmitting station from round~$t$ through round~$t'-1$.
Similarly, station~$r$ cannot distinguish between the first and the third scenario, while station~$s$ cannot distinguish between the second and the third scenario up to round~$t'$.
Therefore the behavior of $r$ at round~$t'$ is the same as under the first scenario while the behavior of~$s$ is the same as under the second scenario: each of the stations~$r$ and~$s$ transmits at round~$t'$ under the third scenario.
This creates a collision at round~$t'$ making it void.

We have exhausted all the possible cases and showed that a void round~$t'$ after round~$t$ has to occur under some scenario.
The behavior of the adversary is as follows.
Let the injections from round~$t$ up to round~$t'-1$ be such that they create a void round~$t'$ after~$t$ the soonest.
Define $t_{i+1}$ to be this void round~$t'$.
If no packets are injected at round~$t'-1$ then let the adversary inject one packet into some station at round~$t'$, otherwise, with two packets injected at round~$t'-1$, let the adversary does not inject any packet at round~$t'$.
This pattern of injections is consistent with the definition of an adversary of burstiness~$2$.
\end{proof}

Theorem~\ref{thm:impossible-retaining-stable} is immediately applicable to algorithms that are fair or that withhold the channel, as they are retaining:


\begin{corollary}
\label{cor:impossible-fair-stable-4-stations-window}

No algorithm that is fair can also be stable in a system of at least four stations against the window adversary of burstiness~$2$.
\end{corollary}

We showed in Theorem~\ref{thm:impossible-3-stations-withhold-stable} that no algorithm that withholds the channel can be stable for three stations against the leaky-bucket  adversary of burstiness~$2$.
Now we can strengthen this to window adversaries, as long as there are at least four stations in the system:


\begin{corollary}
\label{cor:impossible-withholds-channel-4-stations-window}

No algorithm that withholds the channel can be stable in a system of at least four stations against the window adversary of burstiness~$2$.
\end{corollary}


\Paragraph{Queue-size oblivious algorithms.}

We show now that queue-size oblivious algorithms cannot be stable.
The idea of proof is to reduce the question of stability of such algorithms to that of stability of retaining algorithms.
Let us recall that an algorithm is queue-size oblivious if it has two properties: the size of the queue at a node does not affect the decision if to transmit the pending packet at the next round and stability is not affected when an adversary can additionally disturb a stable execution by repeatedly removing a packet from some queue to inject it anywhere at the next round.
Queue-size oblivious algorithms may use the queue size encoded by control bits attached to packets to inform other stations about it, but, unlike algorithm \textsc{Move-Big-To-Front}, cannot have that size affect deciding whether to transmit the currently pending packet at the next round.
In what follows we describe a reduction of stability of queue-size oblivious algorithms to stability for retaining algorithms.

We define a transformation among queue-size oblivious algorithms which, for a given queue-size oblivious algorithm~$\cP$, determines an algorithm~$\cP_h$ that mostly behaves like~$\cP$ and is retaining.
The idea is to simulate~$\cP$ except for some ``reserved'' rounds when $\cP$ comes to rest while packets are transmitted independently of~$\cP$.
A station~$p$ that is heard on the channel at a round~$t$ \emph{reserves the channel for round~$t'>t$} when the message of~$p$ carries control bits that are interpreted by all the stations so that $p$ will be the only station transmitting at round~$t'$.

Round reservations are performed as follows.
When a station~$p$ transmits a packet at a round~$t$ and $p$ has more packets in its queue, then $p$ attaches control bits to the transmitted packet to reserve a round.
The \emph{first available round} for such a round~$t$ is defined to be the first round after~$t$ that is currently not reserved.
The \emph{second available round} is defined similarly as the earliest round after the first available round that is currently not reserved.
The idea is to reserve the second available round.
When a station~$p$ attempts to transmit at a round, then simultaneously $p$ wants to perform one of the following possible actions, which is encoded by control bits attached to the transmitted packet:
\begin{enumerate}
\item[(1)]
if $p$ has at least one packet in its queue then $p$ reserves the second available round;

\item[(2)]
if $p$ has an empty queue and $p$ has a round reserved already, then $p$ cancels the reservation;

\item[(3)]
if $p$ has an empty queue and $p$ does not have a round reserved yet, then $p$ does not reserve nor cancel any round while transmitting its pending packet.
\end{enumerate}
A reservation by a station replaces a previous reservation by the station that is on record.
An action among the three stipulated above is considered to be performed when the transmission is heard, otherwise no reservation is made nor canceled by colliding transmissions.
Observe that each station has at most one round reserved at all times.

To record reservations, each station maintains an array of future rounds reserved by all the stations.
These arrays are identical at all the stations, as updates of entries are performed in the same way by all the stations.
The array is indexed by the names of the stations: an entry indicates how many rounds still remain to have the round reserved by the corresponding station.
At the end of a round the entries are updated in a self-evident manner, depending on the contents of the message heard at the round, if any.

A round that has been reserved  is simply called \emph{reserved} and otherwise it is called  \emph{regular}.
Reservations can be made at any rounds, whether reserved or regular.
The simulation of $\cP$ by~$\cP_{h}$ proceeds at regular rounds, in the sense that decisions inherent to $\cP$ about broadcasting are made at those rounds.
The states of~$\cP_h$ are augmented by additional information related to the mechanism of simulation.
In particular, a station maintains an array for round reservations and an additional buffer space distinct from its queue to temporarily store some of the injected packets.
A station may sometimes use a ``dummy'' packet as pending when no real packet is available at the station.

On the abstract level, the states of $\cP_h$ can be visualized as pairs $(a,b)$, where the \emph{$\cP$-state component~$a$} is the state of~$\cP$ at the current round and the \emph{reservation component~$b$}  is used to represent reservations of rounds and the related additional activity.
The state transitions of~$\cP_h$ on the $\cP$-state components are to implement the functionality of~$\cP$.
Next we describe the lower level of simulation in terms of actions performed by stations.
These actions depend on whether the round is regular or reserved.

\noindent
The case of a regular round~$t$:

If a message is to be transmitted by some station~$p$ running~$\cP$ at round~$t$, then this message is structured according to the specification of $\cP$, possibly including control bits as required by~$\cP$, and additionally may carry control bits for round reservations.
A station hearing the message uses the control bits used by~$\cP$ for the state transition: the new $\cP$-state is obtained by performing a state transition of~$\cP$.
The packets stored at this point in the additional buffer are enqueued in the queue, as if injected by the adversary at round~$t$, followed by enqueueing the packets that were actually injected at round~$t$.

\noindent
The case of a round~$t$ reserved by some station~$p$:

The message transmitted by $p$ at this round consists of the pending packet and control bits for round reservation.
A station hearing the message uses the control bits to update round reservation by~$p$.
Packets injected at this round are stored in the additional buffer.
Station~$p$ discards the transmitted packet and replaces the pending packet by another one obtained as follows.
If the queue is nonempty, then a packet is obtained by dequeueing the queue and made pending, otherwise, if the additional buffer in non-empty, then a packet is removed from it and made pending, otherwise, if there is no packet available at~$p$, then $p$ creates a ``dummy'' packet and it is considered as pending.
Whether the new  pending packet is real or dummy will not affect the state transitions of~$\cP$, as a pending packet itself does not contribute to the state but rather it is the information whether a pending packet exists or not at a round.
A dummy packet is a temporary device, it is to be replaced by a real packet as soon as one becomes injected.
When a dummy packet is transmitted by~$p$ and heard on the channel, it is discarded and not immediately replaced by another dummy packet by~$p$, as no round reservation is made with such a transmission.
Next let us clarify the state transition at a reserved round on the higher level.
The state at a round contains the $\cP$-state component, which we want to stay the same  through the round, except for the queue of the transmitting station~$p$. 
When a new pending packet of~$p$ is obtained by dequeuing the queue, then the queue   contains one element less afterwards.
What is modified in all the stations occurs on the reservation component, which includes reserving rounds, injections of new packets along with the needed manipulations of the additional buffer, and creating a dummy packet if needed.

This completes the specification of the simulating algorithm~$\cP_{h}$.


\begin{lemma}
\label{lem:infinitely-regular-rounds}

There are infinitely many regular rounds in any execution of~$\cP_{h}$, for any input  algorithm~$\cP$.
\end{lemma}

\begin{proof}
A station always reserves the second available round, the first available one left as the next expected regular round.
A change in reservation depends on if it is a new reservation or a cancellation.
A new reservation keeps the first available round intact.
A cancellation of a reservation may release a round before the current first available round, with the effect of making a regular round happen earlier.
\end{proof}


\begin{lemma}
\label{lem:h-transformation}

Let $\cP$ be  a queue-size oblivious algorithm that is stable in a given system of $n$ stations against the window adversary of burstiness~$2$.
Then $\cP_h$ is stable in the same system of $n$ stations against the window adversary of burstiness~$2$.
\end{lemma}

\begin{proof}
Consider an execution~$\cE_{h}$ of~$\cP_h$ with injections determined by the window adversary of burstiness~$2$.
We will prune this execution of the reserved rounds to obtain an execution~$\cE$ of~$\cP$.
While pruning $\cE_{h}$ of reserved rounds, we specify the behavior of the adversary that injects packets only during the regular rounds and this determines execution~$\cE$.
By Lemma~\ref{lem:infinitely-regular-rounds}, there are infinitely many regular rounds in $\cE_{h}$ that allow for~$\cE$ we define next to be unbounded and so legitimate.

We proceed through all the reserved rounds~$t$, one by one in the order of time.
Let some station~$p$ transmit a packet at a reserved round~$t$ of~$\cE_h$.
When the round is removed from the execution, then station~$p$ has one packet less. 
This does not affect the decision of the transmitting station~$p$ to transmit its currently pending packet, if any, at the next round, by the definition of a queue-size oblivious algorithm.
If some $x>0$ packets are injected at this round~$t$ by the adversary,  then let the adversary inject these $x$~packets into the same stations at the round after the just removed one along with the packets to be injected at the round~$t+1$ of execution~$\cE_{h}$.
One of these injected packets, if any are injected, may be interpreted as corresponding to the packet removed from the queue of~$p$ just before the round numbered~$t+1$ in~$\cE_h$ starts in~$\cE$, while the burstiness accounts for the remaining packets.
This does not affect stability by the definition of a queue-size oblivious algorithm.

After removing all the reserved rounds in this way, what is obtained is an execution~$\cE$ of~$\cP$ against a window adversary of the same burstiness.
The execution can be interpreted as resulting from the injections of the adversary that created execution~$\cE_h$ of~$\cP_h$, while the adversary can additionally remove a packet from a queue at some rounds to possibly inject the same packet at the next round.
This means that the sizes of queues depend on successful transmissions in both $\cE$ and~$\cE_h$ in a similar way, and hence that $\cE$ is stable if $\cE_h$ is such.
\end{proof}


\begin{lemma}
\label{lem:retaining}

Algorithm~$\cP_{h}$ is retaining, for any input algorithm $\cP$.
\end{lemma}

\begin{proof}
Suppose that station~$p$ transmits at a round~$t$ and that afterwards the adversary does not inject any packet into~$p$.
Station~$p$ performs a round reservation by way of the message transmitted at round~$t$, unless the queue is empty.
The station keeps reserving the channel in each subsequent transmission at a reserved round, for as long as the queue is nonempty. 
The mechanism of reservation provides that all the packets at the station get transmitted eventually.
\end{proof}


\begin{corollary}
\label{cor:impossible-queue-size-oblivious-stable}

No queue-size oblivious algorithm can be stable in a system of at least four stations against the   window adversary of burstiness~$2$.
\end{corollary}

\begin{proof}
Let $\cP$ be a queue-size oblivious algorithm that  is stable in a given system of $n\ge 4$ stations against the window adversary of burstiness~$2$.
Then, by Lemma~\ref{lem:h-transformation}, algorithm~$\cP_h$ is also stable in the same system of $n$ stations against the window adversary of burstiness~$2$.
This contradicts Theorem~\ref{thm:impossible-retaining-stable}, because algorithm~$\cP_h$ is retaining by Lemma~\ref{lem:retaining}.
\end{proof}

Algorithm \textsc{Move-Big-To-Front} makes a station inform other stations if it has at least $n$ packets available, which grants the station the right to transmit again at the next round.
It follows from Corollary~\ref{cor:impossible-queue-size-oblivious-stable} that any algorithm stable in large systems needs to resort to a similar mechanism of deciding on transmissions based on the sizes of queues.

\section{Conclusion}

\label{sec:conclusion}

We studied deterministic broadcasting algorithms for multiple access channels that need to handle the injection rate of one packet per round.
The question we addressed was what quality of service could be obtained along with throughput~$1$.
We attempted to provide a comprehensive picture of the problem.
We emphasized two aspects.
One of them concerns the power of acknowledgment based algorithms and full sensing ones.
These subclasses of algorithms are natural to define and their namesakes play a prominent role among randomized algorithms.
The other aspect concerns a comparison of the environments determined by window adversaries and leaky-bucket ones.
It turns out that a combination of an adversarial model and a subclass of algorithms determines a unique quality of service that can be achieved in the respective broadcasting environment, depending on the number of stations.
Systems of surprisingly small sizes of just two or three stations are sufficient for  differences in the associated quality of service to be manifested.

The  subclasses of acknowledgment based and full sensing deterministic algorithms we consider are motivated by the corresponding randomized classes of algorithms, see~\cite{Gallager85}.
Historically, investigating algorithmic aspects of broadcasting on multiple access channels has concentrated on randomized algorithms in environments determined by  stochastic constraints on injection rates, with stability understood as either ergodicity of a Markov Chain representing the broadcast environment or as having throughput equal to the injection rate.
The goal was to find the maximum rate for which stability or bounded latency was achievable.
This paper considers deterministic algorithms in adversarial settings.
We show that, as far as the maximum throughput~$1$ is concerned, acknowledgment based algorithms and full sensing ones cannot achieve much, except for systems of just a few stations.
On the other hand, we show that a general deterministic algorithm can achieve the maximum throughput~$1$, although no such an algorithm can simultaneously guarantee bounded latency.

How to compare these results to those in the literature about randomized algorithms, including environments with stochastic constraints on injections?
The authors of this paper do not know any prior work that shows that a randomized algorithm under stochastic assumptions on dynamic packet generation has throughput~$1$.
As randomness can be understood as an additional resource to help an algorithm, it is  surprising that a deterministic algorithm can achieve throughput~$1$ in the worst case.
Apparently stability in the stochastic sense versus stability in the adversarial sense capture different phenomena and stem from different intuitions related to a vague idea of quality of service provided by bounded queues.


\bibliographystyle{plain}

\bibliography{bogdan,books,networks}

\end{document}